\documentclass[aps,amsfonts,amssymb,10pt,letterpaper, 
final, nofootinbib, longbibliography,
twocolumn,notitlepage,hyperref,twoside,allowtoday,noarxiv,accepted=2021-03-17]{quantumarticle}

\usepackage[showerrors,immediate]{silence}
\usepackage[utf8]{inputenx}
\usepackage[OT1]{fontenc}
\usepackage[english]{babel}
\usepackage[style=american,autopunct=true]{csquotes}
\usepackage{soul}

\usepackage{bold-extra}
\usepackage{appendix}

\usepackage{dsfont}

\usepackage{amsmath, amssymb, amsthm,verbatim, graphicx,bbm}
\DeclareMathAlphabet\mathbfcal{OMS}{cmsy}{b}{n}
\usepackage{epsfig}
\usepackage{textcomp}
\usepackage{mathrsfs, mathtools}
\usepackage{units}
\usepackage{fullpage, stackrel, subfigure}
\usepackage{paralist}
\usepackage{comment}
\usepackage{placeins}

\usepackage{subfigure}
\usepackage{mathrsfs}
\usepackage[all]{xy}

\newcommand{\too}{\! \! \to \! \!}

\usepackage{tikz}

\usetikzlibrary{calc,decorations.pathreplacing, intersections,through, patterns}
\usetikzlibrary{arrows,shapes}
\usetikzlibrary{calc,decorations.pathreplacing}
\definecolor{darkgreen}{rgb}{0,.5,0}

\usetikzlibrary{backgrounds}
\pgfdeclarelayer{myback}
\pgfsetlayers{background,myback,main} 

\usepackage{hyperref}

\usepackage{tabularx}
\usepackage{booktabs} 
\newcolumntype{R}{>{\raggedleft\arraybackslash}X}
\newcolumntype{C}{>{\centering\arraybackslash}X}
\newcolumntype{L}{>{\raggedright\arraybackslash}X}

\usepackage{tikz}
\usetikzlibrary{calc,decorations.pathreplacing,decorations.markings}
\usetikzlibrary{arrows,shapes}
\usepackage{accents}

\pgfkeys{/tikz/.cd,
  circle color/.initial=black,
  circle color/.get=\circlecolor,
  circle color/.store in=\circlecolor,
}

\tikzset{dotted pattern/.style args={#1 and #2}{
   postaction=decorate,
   decoration={
    markings,
    mark=
    between positions 0 and 1 step #2
      with
      {
       \fill[radius=#1,\circlecolor] (0,0) circle;
      }
    }
  },
  dotted pattern/.default={1pt and 1.5mm},
}

\newcommand{\ket}[1]{\left| #1 \right>} 

\newcommand{\beq}{\begin{equation}}
\newcommand{\eeq}{\end{equation}}

\newcommand{\LOSR}[0]{\ifmmode\textup{\upshape LOSR}\else{\textup{\upshape LOSR}}\fi}
\newcommand{\LOSE}[0]{\ifmmode\textup{\upshape LOSE}\else{\textup{\upshape LOSE}}\fi}
\newcommand{\DetLOSR}[0]{\ifmmode\textup{\upshape LDO}\else{\textup{\upshape LDO}}\fi}
\newcommand{\LDO}[0]{\ifmmode\textup{\upshape LDO}\else{\textup{\upshape LDO}}\fi}
\newcommand{\LSO}[0]{\ifmmode\textup{\upshape LSO}\else{\textup{\upshape LSO}}\fi}
\newcommand{\LDTNO}[0]{\ifmmode\textup{\upshape LDTNO}\else{\textup{\upshape LDTNO}}\fi}
\newcommand{\LO}[0]{\ifmmode\textup{\upshape LO}\else{\textup{\upshape LO}}\fi}
\newcommand{\LOCC}[0]{\ifmmode\textup{\upshape LOCC}\else{\textup{\upshape LOCC}}\fi}
\newcommand{\NS}[0]{\ifmmode\textup{\upshape NS}\else{\textup{\upshape NS}}\fi}
\newcommand{\SEP}[0]{\ifmmode\textup{\upshape SEP}\else{\textup{\upshape SEP}}\fi}
\newcommand{\GPTCC}[0]{\ifmmode\textup{\upshape GPTCC}\else{\textup{\upshape GPTCC}}\fi}

\newcommand{\nhphantom}[1]{\sbox0{#1}\hspace{-\the\wd0}}

\newcommand*{\LOSEconv}{\xmapsto{\LOSE}}

\usepackage[only,Arrownot]{stmaryrd}

\newcommand*{\LOSEinterconv}{\xleftrightarrow{\LOSE}}

\usepackage{braket}

\newtheorem{theo}{Theorem}
\newtheorem{thm}[theo]{Theorem}
\newtheorem{opq}{Open Question}

\newtheorem{cor}[theo]{Corollary}

\newtheorem{defn}{Definition}

\theoremstyle{definition} 

\theoremstyle{plain}
\newtheorem{innercustomgeneric}{\customgenericname}
\providecommand{\customgenericname}{}
\newcommand{\newcustomtheorem}[2]{%
  \newenvironment{#1}[1]
  {%
   \renewcommand\customgenericname{#2}%
   \renewcommand\theinnercustomgeneric{##1}%
   \innercustomgeneric
  }
  {\endinnercustomgeneric}
}

\newcustomtheorem{customprop}{Proposition}
\newcustomtheorem{customcor}{Corollary}
\newcustomtheorem{customlem}{Lemma}
\newcustomtheorem{customopq}{Open Question}

\begin{document}
\title{Postquantum common-cause channels: the resource theory of local operations and shared entanglement}
\author{David Schmid}
\affiliation{Perimeter Institute for Theoretical Physics, 31 Caroline St. N, Waterloo, Ontario, N2L 2Y5, Canada}
\affiliation{Institute for Quantum Computing and Dept. of Physics and Astronomy, University of Waterloo, Waterloo, Ontario N2L 3G1, Canada}
\email{dschmid@perimeterinstitute.ca}
\author{Haoxing Du}
\affiliation{Perimeter Institute for Theoretical Physics, 31 Caroline St. N, Waterloo, Ontario, N2L 2Y5, Canada}
\author{Maryam Mudassar}
\affiliation{Perimeter Institute for Theoretical Physics, 31 Caroline St. N, Waterloo, Ontario, N2L 2Y5, Canada}
\author{Ghi Coulter-de Wit}
\affiliation{Perimeter Institute for Theoretical Physics, 31 Caroline St. N, Waterloo, Ontario, N2L 2Y5, Canada}
\author{Denis Rosset}
\affiliation{Perimeter Institute for Theoretical Physics, 31 Caroline St. N, Waterloo, Ontario, N2L 2Y5, Canada}
\author{Matty J. Hoban}
\affiliation{Department of Computing, Goldsmiths, University of London, New Cross, London SE14 6NW, United Kingdom}
\begin{abstract}
We define the type-independent resource theory of local operations and shared entanglement (LOSE). 
This allows us to formally quantify postquantumness in common-cause scenarios such as the Bell scenario.
Any nonsignaling bipartite quantum channel which cannot be generated by LOSE operations requires a {\em postquantum common cause} to generate, and constitutes a valuable resource. Our framework allows LOSE operations that arbitrarily transform between different types of resources, which in turn allows us to undertake a systematic study of the different manifestations of postquantum common causes. 
Only three of these have been previously recognized, namely postquantum correlations, postquantum steering, and `non-localizable' channels, all of which are subsumed as special cases of resources in our framework.
Finally, we prove several fundamental results regarding how the type of a resource determines what conversions into other resources are possible, and also places constraints on the resource's ability to provide an advantage in distributed tasks such as nonlocal games, semiquantum games, steering games, etc. 
 \end{abstract}
\maketitle


\section{Introduction}

In space-like separated experiments,
it is well-known that quantum theory can generate correlations which violate Bell inequalities~\cite{Bellreview},  witnessing the fact that they cannot be explained by a classical common-cause process~\cite{Wood2015}. 
In a common-cause scenario, such correlations can only be explained by a common cause described by a more general theory, such as a quantum common cause~\cite{Allenetal,Lorenz2019,CostaShrapnel} (which for our purposes can be viewed  simply as a composite quantum system in an entangled state). 

However, it is also well-known that in common-cause scenarios, quantum theory cannot generate correlations which are {\em maximally} nonclassical, as was pointed out by Tsirelson~\cite{Cirel'son1980} and by Popescu and Rohrlich~\cite{Popescu1994}. That is, nonsignaling correlations which achieve the logically maximal violation of a Bell inequality generally cannot be achieved by local measurements on entangled states, and yet {\em are} achievable by {\em postquantum} common causes~\cite{barrettGPT}, that is, those described by a generalized probabilistic theory~\cite{hardy01,barrettGPT} (GPT) beyond quantum theory.


Correlations, often termed box-type resources, or simply `boxes'~\cite{wolfe2020quantifying}, are not the only postquantum resources. More general types of postquantum resources\footnote{Throughout this work, the adjective postquantum refers to the common cause required to generate the given resource, while the resources themselves are always channels taking quantum states to quantum states (as opposed to channels which act on postquantum states, which we do not consider).
} have been studied in a few previous works~\cite{causallocaliz,nosigboxes,PerinottiLOSEex,Hoban_2018,BobWI}, which found instances of channels which cannot be generated by local operations and shared entanglement (LOSE)~\cite{Gutoski2008}.

Ref.~\cite{causallocaliz} was the first to show that the set of nonsignaling channels is strictly larger than the set of channels realizable by local operations and shared entanglement. 
The authors refer to the former set of channels as `causal' and to the latter set of channels as `localizable', but we will not adopt this terminology (largely because the term `causal' is ambiguous and overloaded). 
A key motivation in Ref.~\cite{causallocaliz} (which also applies to our work) is to better understand the restrictions on operations that are imposed by special relativity.
Subsequently, Refs.~\cite{nosigboxes,PerinottiLOSEex} made some progress in characterizing the set of postquantum channels, and Ref.~\cite{PerinottiLOSEex} showed that there exist postquantum channels that are not entanglement-breaking.

Additionally, Ref.~\cite{BobWI} studied the set of postquantum resources which arise in generalized steering scenarios in the bipartite case. First, Ref.~\cite{PostquantumSteering} noted that all logically possible bipartite quantum steering assemblages can be generated using local operations and shared entanglement, and hence that postquantum common causes do not give an advantage for steering of quantum states. Initially this motivated the study of multipartite postquantum steering~\cite{PostquantumSteering,LocalMeasurement,Hoban_2018}.  However, returning to the bipartite setting, Ref.~\cite{BobWI} considered a generalized type of steering scenario in which the steered party (Bob) has a classical input system. These were termed Bob-with-input steering assemblages, and Ref.~\cite{BobWI} provides examples of such resources which cannot be generated using local operations and shared entanglement.

In fact, there are a wide variety of nontrivial postquantum resources, including boxes, distributed measurements, and Bob-with-input steering assemblages as special cases, as well as measurement-device-independent steering assemblages, channel-steering assemblages, ensemble-preparing channels, and so on.

We introduce a unified framework for the study of postquantum common cause channels, which subsumes all of these types of resources as special cases. This framework is the resource theory~\cite{resthry} of local operations and shared entanglement. Any common-cause process which cannot be realized by LOSE operations is a valuable resource and a signature of a postquantum common cause. 
Furthermore, by considering transformations between resources that can be enacted using LOSE operations, one can quantitatively characterize the postquantumness of any given resource.

The study of such postquantum resources sheds light on the space of logically conceivable processes which are not realizable with quantum common causes. As evidenced by the example of the Popescu-Rohrlich (PR) box, such processes serve as useful foils~\cite{chiribella_spekkens_2016} for teaching us about quantum theory. Furthermore, postquantum processes provide clues for where one might find physics beyond quantum theory~\cite{hardy2005probability}.
Conversely, a recent line of research involves the search for reasonable physical principles which distinguish nonlocal quantum correlations from postquantum correlations: macroscopic locality~\cite{Navascues2010}, information causality~\cite{Pawlowski2009,Pawlowski2016}, limits on communication complexity~\cite{Brassard2006} or nonlocal computation~\cite{Linden2007, Broadbent2016}, local orthogonality~\cite{Fritz2013a} or the Specker principle~\cite{Cabello2012,Gonda2018}.
Almost quantum correlations~\cite{Navascues2015,Henson2015} were originally developed to provide an example of a postquantum theory that nevertheless satisfies several of the reasonable physical principles above.
This postquantum theory was first studied in the context of nonlocal boxes, and was then extended to generalized steering scenarios~\cite{PostquantumSteering,BobWI} and beyond~\cite{Hoban_2018}.
Discriminating quantum from postquantum (as opposed to discriminating nonclassical from classical) can also be useful in the field of device-independent quantum information~\cite{NPA}, such as for self-testing \cite{branciard, shca}, random certification \cite{Colbeckthesis,Colbeck_2011,pironio, Supic2017}, and the study of nonlocal games~\cite{Wehner}.
We also note recent experimental explorations of the boundary of quantum theory~\cite{Mazurek2019,Liang2019}.

Most of the results in this article refer not to specific instances of postquantum resources, but rather to what can be said about a postquantum resource merely from knowing its type (e.g. whether it is a box, an assemblage, a distributed measurement, etc). Most of our analysis and results have close analogues in Ref.~\cite{schmid2020type}, which undertook a similar study for the type-independent resource theory of local operations and shared randomness (LOSR), which is the appropriate resource theory for quantifying {\em nonclassicality} of common-cause processes~\cite{wolfe2020quantifying} (to be contrasted with the resource theory of LOSE studied here, which is appropriate for quantifying {\em postquantumness} of common-cause processes). Both of these works are grounded conceptually in the idea that correlations should be deemed classical or nonclassical depending on what sort of causal (and inferential) theories are required to reproduce them while respecting the conservative network structure~\cite{Wood2015,wolfe2020quantifying,schmid2020unscrambling}.

In Section~\ref{envelp}, we introduce the types of resources we consider; that is, the enveloping set of resources in our resource theory. 
In Section~\ref{sectypes}, we introduce the free resources and free transformations. We then characterize the types of bipartite resources which can possibly be nonfree---that is, which can exhibit postquantumness. 
In Section~\ref{freetrans}, we discuss how the free transformations induce a preorder over all resources, which quantitatively determines how valuable each resource is.
In Section~\ref{comppq}, we state some results (proved in Appendix~\ref{sec:examples}) about the relative value of various postquantum resources that have been considered in the literature.
In Section~\ref{encodingsec}, we discuss type encodings, which formalize the idea that some resource types are able to exhibit all of the same manifestations of postquantumness as another type of resource. We begin a systematic study of which types of resources can encode the postquantumness of all resources of other types, and we show that the distributed measurement type can encode the postquantumness of {\em all} resources of all types. 
In Section~\ref{unifiedgames}, we recall the framework of distributed games, introduced in Ref.~\cite{schmid2020type}, which unifies and generalizes many traditional distributed tasks and games, including entanglement witnessing, nonlocal games, semiquantum games, steering games, teleportation games, and so on. We then discuss how players can implement type-changing LOSE operations on whatever shared resources they have access to in order to generate an optimal strategy for the game they are playing.
In Section~\ref{implictypetoperf}, we discuss implications from the type of a shared resource to how well it can perform at various games. 
In Section~\ref{secopq}, we conclude with some interesting open questions.

\section{The enveloping theory} \label{envelp}

Following Ref.~\cite{schmid2020type,rosset2019characterizing}, the set of {\bf resources} we consider in this work are the multipartite, completely-positive~\cite{NielsenAndChuang,Schmidcausal} trace preserving (CPTP) maps that are nonsignaling~\cite{rosset2019characterizing} between every pair of parties. We will focus on the bipartite case for simplicity; however, most of our framework and results generalize immediately to arbitrarily many parties.

Different {\bf types} of resources are distinguished by whether the input and output systems of each party are trivial, classical, or quantum. As in Ref.~\cite{rosset2019characterizing}, a single party's system is said to be trivial if it has dimension one, is said to be classical if all operators on its Hilbert space are diagonal, and is otherwise said to be quantum. If a single party has a collection of input (output) systems, we will consider them as a single effective input (output) system.
We will focus on bipartite resources shared between two parties, Alice and Bob, and will denote Alice's input and output systems by $X$ and $A$, respectively, and Bob's input and output systems by $Y$ and $B$, respectively. Each of these systems has trivial ($\mathsf{I}$), classical ($\mathsf{C}$), or quantum ($\mathsf{Q}$) type, and so we can denote the {\bf global type} of the resource by $T := \mathbf{Type}[\mathsf{X}]\mathbf{Type}[\mathsf{Y}] \too \mathbf{Type}[\mathsf{A}] \mathbf{Type}[\mathsf{B}]$. The {\bf partition-type} of a single party is specified by the types of its inputs and outputs, e.g. in Alice's case we write $T_A := \mathbf{Type}[\mathsf{X}] \too \mathbf{Type}[\mathsf{A}]$.

The most commonly studied types of resources, depicted in Figure~\ref{types} and discussed in more detail in Refs.~\cite{schmid2020type,rosset2019characterizing}, are quantum states, of type $\mathsf{II} \too \mathsf{QQ}$, steering assemblages~\cite{Einstein1935, schrodinger_1935,wisesteer, Skrzypczyk2014, Gallego2015, Piani2015a, Cavalcanti2017, Uola2019,Pusey2013}, of type $\mathsf{CI} \too \mathsf{CQ}$, teleportages~\cite{telep, PhysRevA.99.032334,Hoban_2018}, of type $\mathsf{QI} \too \mathsf{CQ}$, nonsignaling boxes~\cite{Barrett2005, brunner2013Bell}, of type $\mathsf{CC} \too \mathsf{CC}$, { distributed measurements} (or { semiquantum channels})~\cite{sq,Supic2017,Hoban_2018}, of type $\mathsf{QQ} \too \mathsf{CC}$, { MDI-steering channels}~\cite{steer}, of type $\mathsf{CQ} \too \mathsf{CC}$, channel steering assemblages~\cite{channelsteer}, of type $\mathsf{CQ} \too \mathsf{CQ}$, { Bob-with-input steering assemblages}~\cite{Hoban_2018,BobWI}, of type $\mathsf{CC} \too \mathsf{CQ}$, { distributed ensemble-preparing channels}~\cite{causallocaliz}, of type $\mathsf{CC} \too \mathsf{QQ}$, and generic nonsignaling bipartite { quantum channels}, of type $\mathsf{QQ} \too \mathsf{QQ}$.
There are also five types of nontrivial bipartite resources that have not (to our knowledge) been previously studied, namely $\mathsf{QC} \too \mathsf{CQ}$, $\mathsf{CQ} \too \mathsf{QQ}$, $\mathsf{IQ} \too \mathsf{QQ}$, $ \mathsf{QQ} \too \mathsf{CQ}$, and $\mathsf{CI} \too \mathsf{QQ}$.

\begin{figure}[htb!]
\centering
\includegraphics[width=0.499\textwidth]{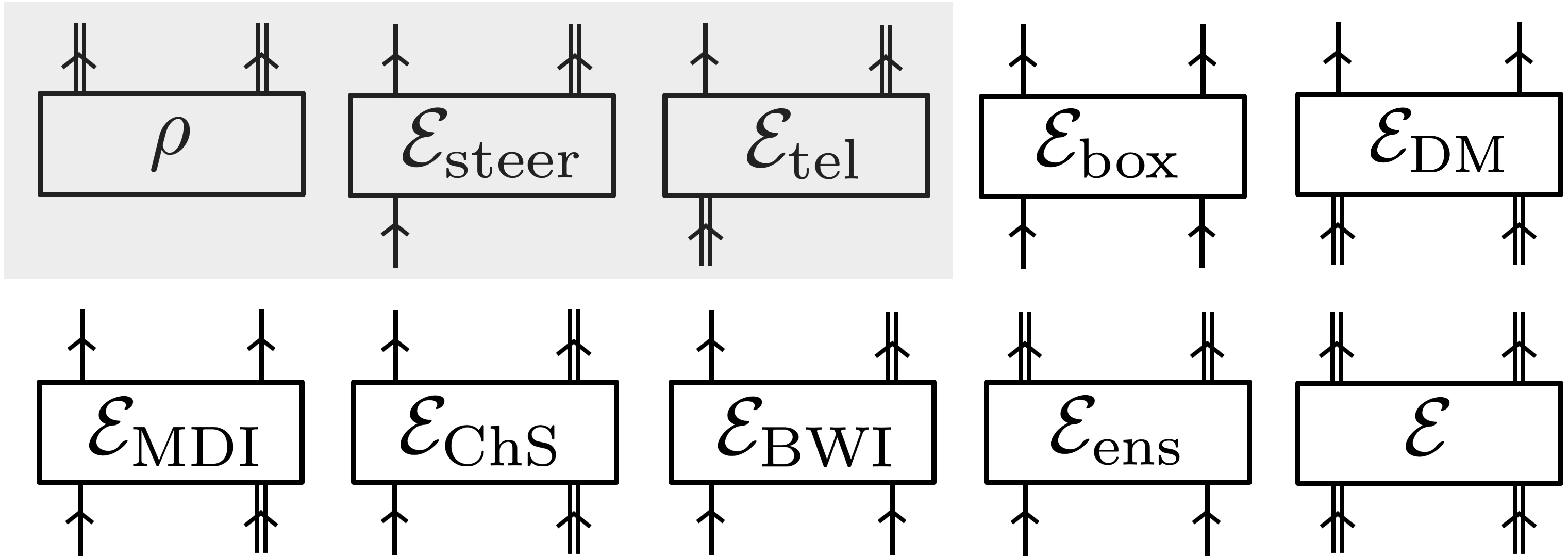}
\caption{Common types of nonsignaling resources, where classical systems are represented by single wires and quantum systems are represented by double wires. As we will show, every resource of the three shaded types is LOSE-free.
From left to right, top to bottom, we depict: quantum states, steering assemblages, teleportages, boxes, distributed measurements (or semiquantum channels), measurement-device-independent steering assemblages, channel steering assemblages, Bob-with-input steering assemblages, distributed ensemble-preparations, and generic nonsignaling bipartite quantum channels. 
}\label{types}
\end{figure}

We denote the set of all nonsignaling resources (of all types) as $\mathbf{E}_{\rm ns}$, and we take this as our enveloping theory.
One could consider a different set of multipartite channels as one's enveloping theory, such as the set $\mathbf{E}_{\rm gcc}$ of multipartite CPTP maps which can be realized by GPT common causes, or the set $\mathbf{E}_{\rm sig}$ of all CPTP maps (including signaling maps). 
Note that $\mathbf{E}_{\rm gcc} \subset \mathbf{E}_{\rm ns} \subset \mathbf{E}_{\rm sig}$. 

It is an interesting open question whether $\mathbf{E}_{\rm gcc}$ is a {\em strict} subset of $\mathbf{E}_{\rm ns}$, or whether the two are equal.
\begin{opq} \label{opq1}
Do there exist bipartite nonsignaling quantum channels which cannot be realized by GPT common causes?
\end{opq}

If the two sets coincide, then one would have a simple characterization of the set of quantum channels realizable by a GPT common-cause, which a priori appears to be a difficult set to characterize. 
For the special case of box-type resources, it was shown in Ref.~\cite{barrettGPT} that GPT common-cause processes {\em can} generate all nonsignaling boxes. On the other hand, it is known that in generic causal scenarios, the set of correlations (over some observed classical variables) that can be generated by unobserved GPT common causes is constrained not only by equality constraints, but also by nontrivial inequality constraints~\cite{Henson2014} (beyond positivity, normalization, and nonsignaling constraints). This makes it plausible that even in the simple common-cause scenario we are considering, the set of processes over quantum systems might also be constrained beyond the nonsignaling constraint.

Apart from its intrinsic interest, characterizing the set $\mathbf{E}_{\rm gcc}$ is worth exploring because it constitutes the most natural choice for an enveloping set of processes in our resource theory. Given a nonsignaling channel that cannot be generated by LOSE operations, there are two possible explanations: either the channel was generated by a postquantum common cause, or it was generated by fine-tuned cause-effect influences.\footnote{The former view is in the spirit of Ref.~\cite{wolfe2020quantifying,Allenetal}, while the latter, endorsed (e.g.) in Ref.~\cite{nosigboxes}, seems to be more traditional.} If the causal structure in a given scenario is a common-cause structure (with no cause-effect relations between the parties), then {\em only} the postquantum common-cause explanation is viable. Such an assumption on the causal structure can be grounded, for example, by relativity theory in the context of space-like separated parties. We will take this approach, and will refer to valuable LOSE channels as {\em resources of postquantum common cause}, or simply `postquantum resources'. As such, the enveloping theory we are really interested in is $\mathbf{E}_{\rm gcc}$. If there do exist any processes in $\mathbf{E}_{\rm ns}$ that are not in $\mathbf{E}_{\rm gcc}$, then these could {\em only} be explained by a resource of fine-tuned cause-effect influence, not by a resource of postquantum common cause, and hence our interpretation of such resources would be inconsistent. As there is no evidence that such processes exist, we will for the purpose of this paper assume they do not, allowing us to take the well-characterized set $\mathbf{E}_{\rm ns}$ as our enveloping theory.\footnote{Subsequent work~\cite{gptccchannels} by Selby et. al. has shown that indeed, no such processes exist, answering Open Question 1 in the negative.}

The third option for the enveloping set, $\mathbf{E}_{\rm sig}$, is also sensible, as it is the most general set of channels that is natural in our context. However, this added generality makes the resulting resource theory more difficult to characterize. In any case, the choice of the enveloping set is not especially consequential, since the value of resources depends only on the free transformations, not on the enveloping set of resources. Hence, we have chosen the set which leads to the most interesting insights, namely $\mathbf{E}_{\rm ns}$.

\section{ LOSE resources and transformations} \label{sectypes}

The key component of any resource theory is a set of {\bf free} processes, that is, resources and transformations which can be implemented at no cost.
Here, we wish to quantify the notion of {\em postquantum common causes} as resources. Consequently, the set of free processes are taken to be all and only those that can be generated by quantum common causes---that is, by shared entanglement. We imagine that all local quantum operations are free, but that all forms of communication between parties are inaccessible. In other words, the free processes are defined by the set of local operations and shared entanglement, or {\bf LOSE operations}. 

 A generic free resource is depicted in Fig.~\ref{freeset2}(a), while free resources of three specific types are shown in Fig.~\ref{freeset2}(b), Fig.~\ref{freeset2}(c), and Fig.~\ref{freeset2}(d).
Any resource in the enveloping theory that {\em cannot} be generated in this manner is {\bf nonfree}: a valuable resource which requires a postquantum common-cause to generate. We provide various explicit examples of valuable resources (of several different types) in Appendix~\ref{sec:examples}.

A generic free transformation is depicted in Fig.~\ref{typechange}. We elaborate on the set of free transformations in Section~\ref{freetrans}, as well as how these free transformations generate the essential structure of our resource theory, that is, the preorder which determines the value of every resource in the enveloping theory. We give three examples of (type-changing) free transformations in Appendix~\ref{sec:examples}.

\begin{figure}[htb!]
\centering
\includegraphics[width=0.499\textwidth]{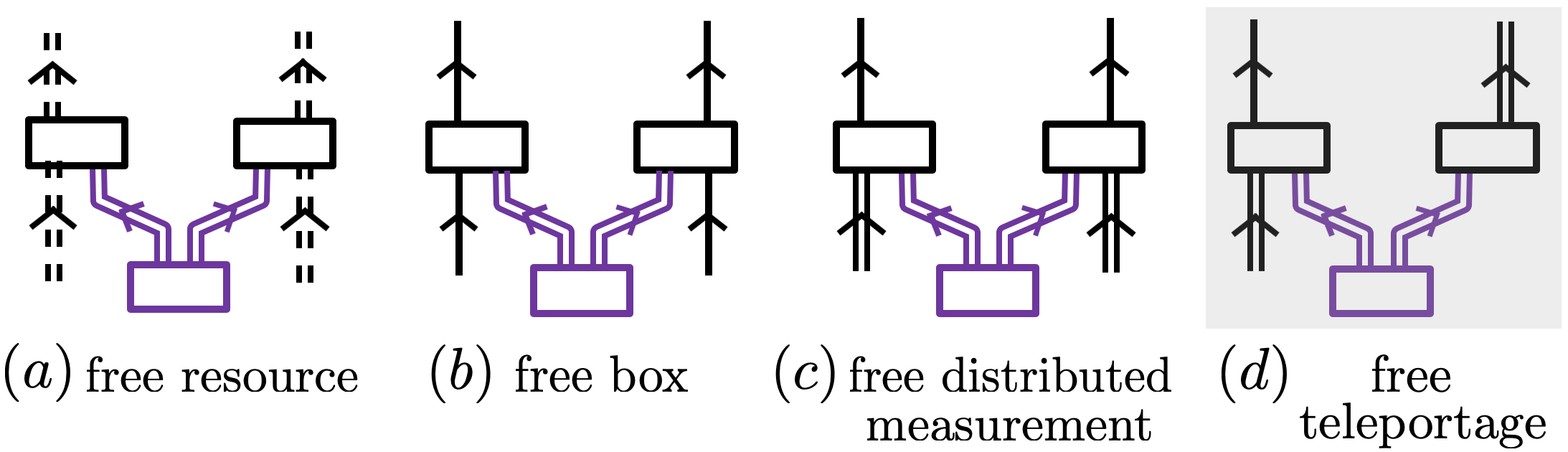}
\caption{Dashed wires depict systems with arbitrary type. A free resource (of arbitrary type) is one which can be generated as in (a), by shared entanglement (in purple) and local operations (in black). Free boxes are those which decompose as in (b), and free distributed measurements are those which decompose as in (c). As a consequence of Theorem~\ref{trivialtypes}, all teleportages are free, since they admit of a decomposition as in (d).}\label{freeset2}
\end{figure}

Many types of resources are trivial, in the sense that quantum theory can generate every resource of that kind, which in turn implies that there are no postquantum resources of that type. 
\begin{defn}
A resource type is {\bf LOSE-trivial} if every resource of that type can be generated by $\LOSE$ operations.
\end{defn}
\noindent Of the resource types in Fig.~\ref{types}, entangled states are all free (by definition), while steering assemblages and teleportages are (less obviously) also all free. We now provide an elegant characterization of the bipartite resource types that can give a postquantum advantage.

\begin{thm} \label{trivialtypes}
A bipartite resource type is LOSE-trivial if and only if it has at least one trivial input or output.
\end{thm}

\begin{proof}
First, note that there exist examples of boxes which are not LOSE-free; a well-known example of such a box is the PR box~\cite{Popescu1994}, defined in Appendix~\ref{sec:examples}. Since every type of bipartite resource whose inputs and outputs are all nontrivial has the PR box as a special case, it follows that none of these types is LOSE-trivial.

Conversely, we prove that every bipartite type which has a trivial input or output is LOSE-trivial. 
The argument relies on the fact, proven in Ref.~\cite{Eggeling_2002} that channels which are nonsignaling from one party to a second can always be realized without any causal influences from the first to the second.
Consider the case where an output of the bipartite channel is trivial; taking this to be Alice's output without loss of generality, the  channel (denoted $\mathcal{E}$) is depicted on the left-hand-side of Fig.~\ref{Eggproof}(a). As we are assuming that $\mathcal{E}$ is nonsignaling from Alice's input to Bob's output, the result of Ref.~\cite{Eggeling_2002} states that this channel can be decomposed into two channels---one of which has no output---as shown. But the only CPTP channel with no outputs is the trace, which factorizes, as shown in the second equality. The resulting channel is manifestly constructed from LOSE operations, and hence $\mathcal{E}$ can be generated by LOSE operations.
Next, consider the case where an input of the bipartite channel is trivial; taking this to be Alice's input without loss of generality, the channel (denoted $\mathcal{E}'$) is depicted on the left-hand-side of Fig.~\ref{Eggproof}(b). As we are assuming that $\mathcal{E}'$ is nonsignaling from Bob's input to Alice's output, the result of Ref.~\cite{Eggeling_2002} states that this channel can be decomposed as shown. But this decomposition requires only LOSE operations, and hence $\mathcal{E}'$ can also be generated by LOSE operations.
\begin{figure}[htb!]
\centering
\includegraphics[width=0.4\textwidth]{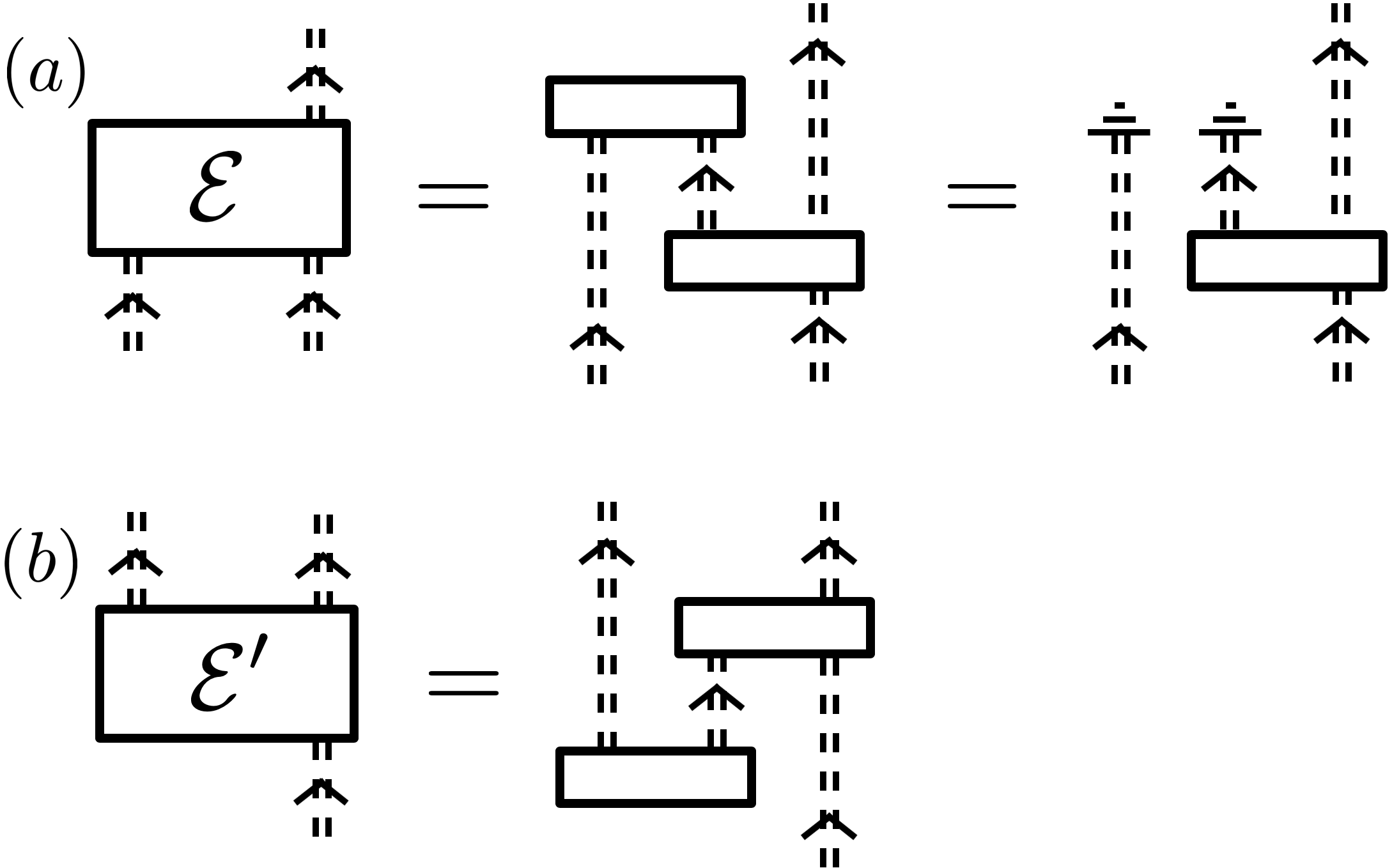}
\caption{ 
(a) Every channel $\mathcal{E}$ of the form on the left-hand-side which is nonsignaling from Alice to Bob can be decomposed into LOSE operations, as on the right-hand-side.
(b)  Every channel $\mathcal{E}'$ of the form on the left-hand-side which is nonsignaling from Bob to Alice can be decomposed into LOSE operations, as on the right-hand-side.
} \label{Eggproof}
\end{figure}
\noindent Since this sequence of operations is manifestly LOSE-free, it follows that $\mathcal{E}$ is necessarily LOSE-free, merely by virtue of its type.

\end{proof}

Theorem~\ref{trivialtypes} has the following corollary, which is formalized in Section~\ref{implictypetoperf}.

\begin{cor} \label{nopqadv}
No GPT can outperform quantum theory at steering or teleporting quantum states; that is, at any steering games~\cite{schmid2020type,Cavalcanti2017} or teleportation games~\cite{schmid2020type,telep,LipkaBartosik2019}.
\end{cor}

  \subsection{Type-changing $\LOSE$ operations} \label{freetrans}

In keeping with the discussion above, the set of free transformations are all and only those that can be generated by LOSE operations. We denote this set by $\LOSE$, an element of this set by $\tau \in \LOSE$, and a generic resource of arbitrary type by $R$, where the action of $\tau$ on $R$ is denoted $\tau \circ R$.

Using the fact that the most general map taking a quantum channel to a quantum channel is a {\em quantum comb}~\cite{qcombs08,qcombs09}, we can decompose an arbitrary (bipartite) free transformation as in Fig.~\ref{typechange}. That is, any free transformation can be generated by each party locally implementing some quantum comb, where furthermore each local comb has as input one subsystem of an arbitrary shared entangled state. (Note that this entangled state can also be fed down the side channel of each local comb, and hence is effectively an input to both the pre-processings and the post-processings that constitute the local comb.) The set of LOSE operations is manifestly convex.

\begin{figure}[htb!]
\centering
\includegraphics[width=0.3\textwidth]{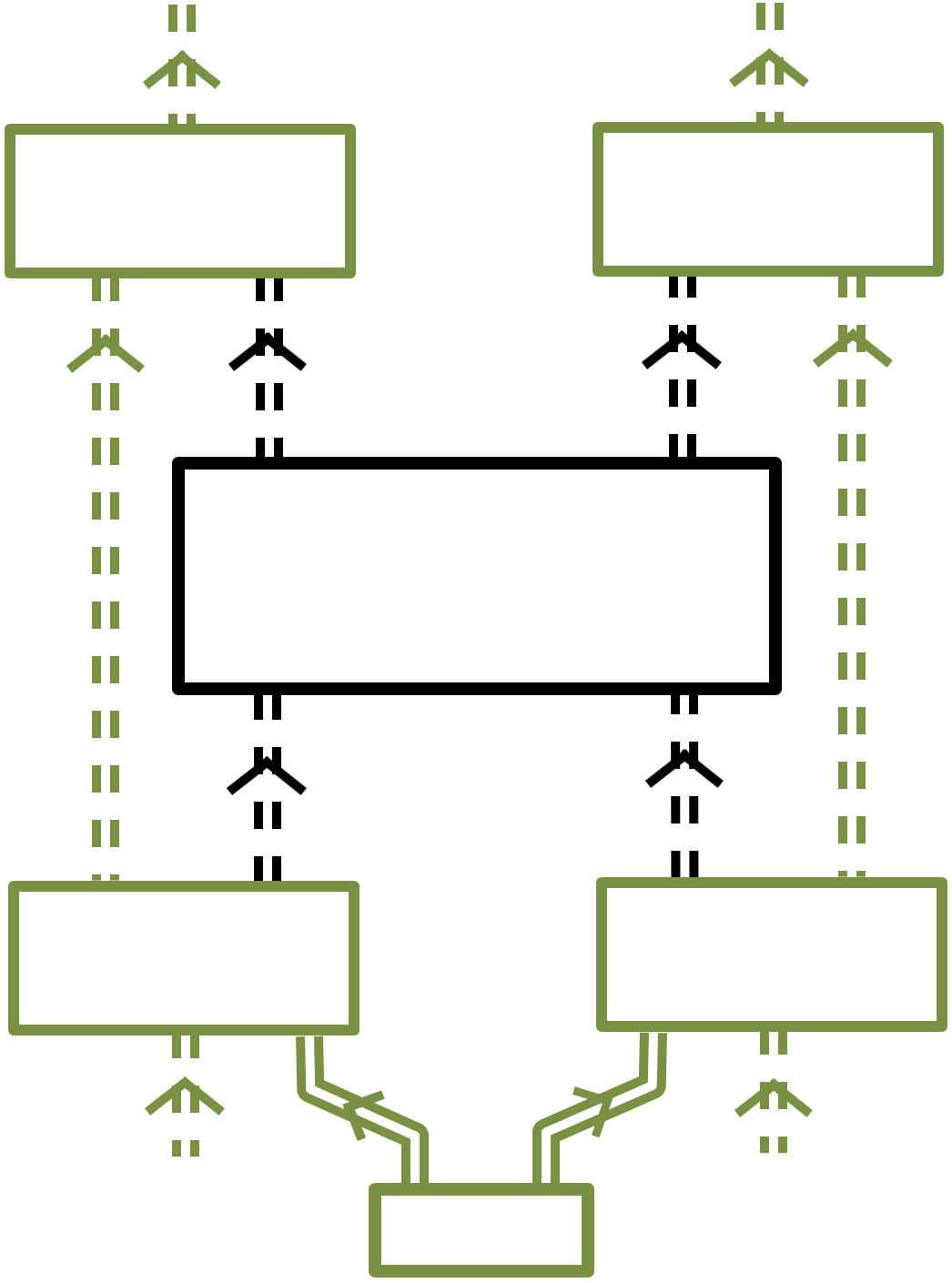}
\caption{ 
A generic $\LOSE$ transformation (in green) on a bipartite resource (in black) consists of a local comb for each party, where each local comb may have as input one subsystem of an arbitrary shared entangled state.
} \label{typechange}
\end{figure}

Critically, note that the local combs can change the type of a resource. That is, the types of the input and output systems of the transformed resource (in green in Fig.~\ref{typechange}) need not be the same types as the input and output systems of the source resource (in black in Fig.~\ref{typechange}). The type of systems of the source and target resource merely determine what sort of pre-and-post-processings make up the local combs.
We give three specific examples of type-changing transformations in Appendix~\ref{sec:examples}; in particular, see Fig.~\ref{Piani}, Fig.~\ref{PRtoBWI}, Fig.~\ref{PHHHtoDFP}, and the surrounding discussions.
 
This set of free transformations induces a preorder over the set of resources in the enveloping theory. 
\begin{defn} \label{LOSEorder}
For resources $R$ and $R'$ of two arbitrary types, we say that $R \succeq_{\rm LOSE} R'$, or equivalently that $R \LOSEconv R'$, iff there exists some $\tau \in \LOSE$ such that $R'= \tau \circ R$.
\end{defn}
\noindent In such cases, we say that $R$ is {\bf at least as postquantum} (or `as valuable') as $R'$. 
This is indeed a preorder, since the identity transformation is LOSE-free (reflexivity) and since $\LOSE$ transformations are closed under composition (transitivity).

Resources $R$ and $R'$ are {\bf equally postquantum} if they are freely interconvertible under $\LOSE$ transformations; that is, if $R \LOSEconv R'$ and $R' \LOSEconv R$. We denote this $R \LOSEinterconv R'$ and say that $R$ and $R'$ are in the same $\LOSE$ equivalence class.

Resources $R$ and $R'$ are {\bf incomparable} if neither can convert to the other under $\LOSE$ transformations.

\subsection{Comparing postquantumness across resource types} \label{comppq}
With this resource theory framework in hand, one can go about quantitatively comparing the postquantumness of resources, even across different types. In this section, we discuss what can be said about the relative value of some known examples of postquantum resources (of various different types) from the literature. 

The resources we discuss here are defined explicitly in Appendix~\ref{sec:examples}. There are five of these, of four different types, and we will refer to them based on their type and the initials of the authors who first considered them.
These include the {\em PR box}~\cite{Popescu1994}, the {\em PHHH ensemble-preparation}~\cite{nosigboxes}, the {\em SHSA Bob-with-input steering assemblage}~\cite{BobWI},
and two nonsignaling channels of type $\mathsf{QQ} \too \mathsf{QQ}$, namely the {\em BGNP channel}~\cite{causallocaliz} and the {\em DFP channel}~\cite{PerinottiLOSEex}. 
In Appendix~\ref{sec:examples}, we prove that four conversions between these resources are possible using LOSE operations, which in turn teaches us about their relative postquantumness. For instance, we show that the PR box is interconvertible with the PHHH ensemble-preparation; hence, the two are equally postquantum. 
Additionally, we show that the PHHH ensemble-preparation can be converted into the SHSA Bob-with-input-assemblage and the DFP channel. By transitivity, of course, this also implies that the PR box can be converted into either of these resources.

  These results (and some remaining open questions) are summarized in Fig.~\ref{knownconv}. 

\begin{figure}[htb!]
\centering
\includegraphics[width=0.48\textwidth]{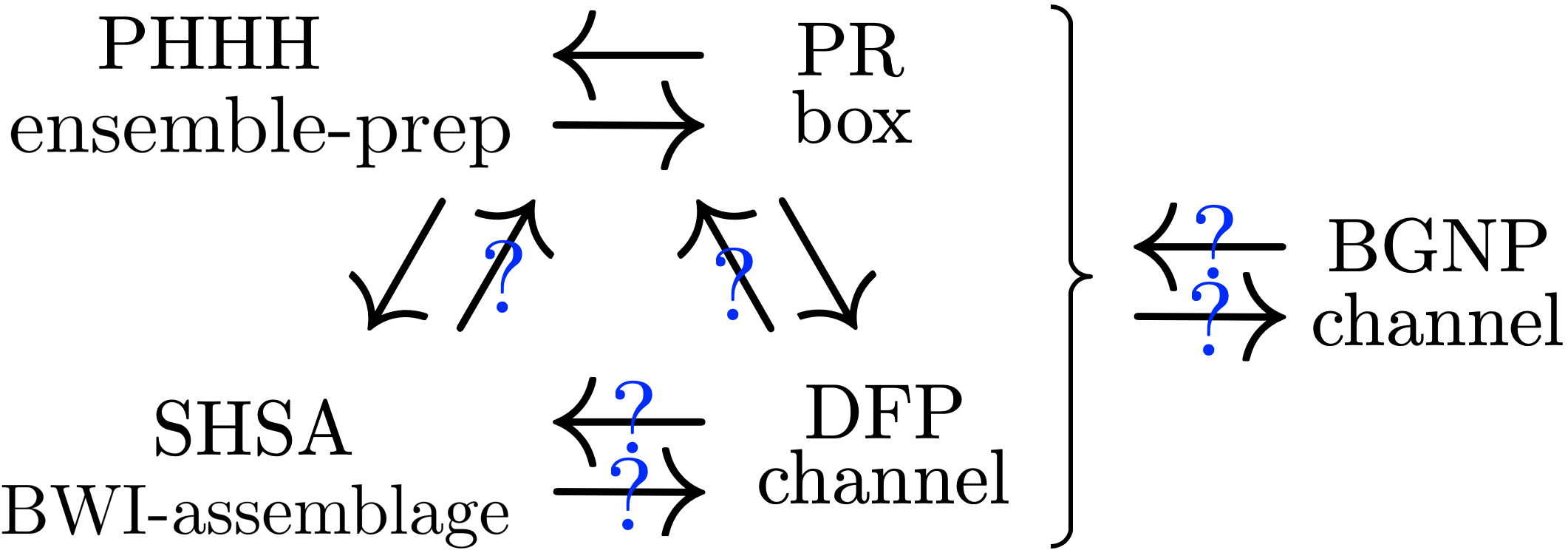}
\caption{ 
We prove that the PR box and the PHHH ensemble-preparation are interconvertible, and either of these can be converted into the SHSA Bob-with-input assemblage or into the DFP assemblage. The remaining conversion relations are unknown (although see the main text for some relevant facts).} \label{knownconv}
\end{figure}

We were not able to resolve the remaining conversion relations, indicated by the blue question marks in the figure. However, there are some hints as to whether or not some of these conversions are possible.
For example, the form of the BGNP channel\footnote{E.g., the fact that the dephasing that is applied must be identical in three subspaces but distinct in the fourth.} suggests that it can also be obtained from the PR box. Also, it was shown in Ref.~\cite{PerinottiLOSEex} that the DFP channel exhibits a non-maximal violation of the Clauser-Horne-Shimony-Holt (CHSH) inequality, which suggests that one might be able to construct a monotone from the CHSH functional (analogous to the constructions in Ref.~\cite{wolfe2020quantifying}), such that the value achieved by the PR box is higher than that achieved by the DFP channel. If this were the case, it would establish that the DFP channel could not be converted to the PR box, which in turn would prove that the PR box is {\em strictly} more postquantum. Finally, we note that  Ref.~\cite{BobWI} establishes that the SHSA Bob-with-input assemblage cannot be converted to any postquantum box (including the PR box) using operations which can be extended to include all local operations and shared randomness. This implies that, if the SHSA Bob-with-input assemblage can be converted to the PR box, the conversion will definitely require shared entanglement.

These results give us some insight into the structure of the type-independent resource theory of LOSE. First of all, it shows that the property of entanglement-breaking is not critical for postquantumness: the PR box, despite being completely entanglement-breaking, is at least as postquantum as the DFP channel, which is not entanglement-breaking. Additionally, we see that there are no examples of postquantum channels that have been proven to be inequivalent to the PR box, nor that have been proven to be strictly more resourceful than the PR box. 
Hence, it remains a possibility that {\em all} postquantum channels can be generated from box-type postquantum channels. This would be a striking result, implying that all manifestations of postquantumness in nature are extremely simple---and limited. It would also imply that the box-type encodes all other types, according to the definition of encoding in the next section. 
While we suspect that not all postquantum channels can be generated from postquantum boxes, this remains a critical open question. 
\begin{opq}
Do there exist LOSE-nonfree channels which cannot be freely generated using LOSE operations and arbitrary nonsignaling boxes?
\end{opq}

Finally, we learn that resources of more limited type (like boxes) can sometimes exhibit more postquantumness than resources of a more general type (like generic quantum channels). We elaborate on this in the next section.

\section{Encoding postquantumness of one type into postquantumness of another type} \label{encodingsec}

It is interesting to consider transformations from one type of resource to another which do not degrade the value of the initial resource. We will prove some general results that do not rely on the particular details of the initial resource, but depend only on its type. To do so, we require the following definition, which defines a preorder over types.
\begin{defn}
Global type $T$ {\bf encodes} global type $T'$, denoted $T \succeq_{\rm type} T'$, if for every resource $R'$ of type $T'$, there exists at least one resource $R$ of type $T$ such that $R \LOSEinterconv R'$.
\end{defn}
\noindent In other words, there exists at least one resource of the higher type in every equivalence class of resources of the lower type. The higher type is able to exhibit all the same manifestations of postquantum common-cause resourcefulness as the lower type. Furthermore, we will show that some resource of the higher type is always able to perform at least as well as any given resource of the lower type at any distributed game.

\subsection{Semiquantum channels encode all types of postquantumness} \label{sec:sq}

We now show that resources with only classical outputs are able to encode the postquantumness of all other types of resources. We focus on the bipartite case for simplicity, but the result generalizes immediately to arbitrarily many parties.

\begin{thm} \label{squniv}
Every bipartite resource (of arbitrary type) can be converted to a bipartite distributed measurement (of type $\mathsf{QQ} \too \mathsf{CC}$) which is in the same equivalence class.
\end{thm}
\noindent The proof is constructive, and is given in Appendix~\ref{proof1}. Furthermore, one can immediately see from the proof that the transformation from the given resource to a distributed measurement in its equivalence class can be achieved by a single fixed transformation, independent of the specific initial resource. 

This fact implies that all resources of arbitrary types can have their postquantumness quantified in a measurement-device-independent~\cite{steer} fashion, since classical systems may be probed even without well-characterized quantum measurement devices.

\subsection{Partition-type encodings}

There are a large variety of global types, especially in multipartite scenarios. A useful tool for studying the preorder over these is the notion of a preorder over partition-types, that is, over the nine types $T_A := \mathbf{Type}[\mathsf{X}] \too \mathbf{Type}[\mathsf{A}]$ of (e.g.) Alice's share of a resource. 
\begin{defn} \label{typeorder}
Partition-type $T_1$ is above partition-type $T_1'$, denoted $T_1 \succeq_{\rm type} T_1'$, if for every resource $R'$ of type $T_1'  T_2 ...  T_n $ (as one ranges over all $T_2, ..., T_n$ and all $n$), there exists a resource $R$ of type $T_1  T_2 ...  T_n $ satisfying $R' \LOSEinterconv R $. 
\end{defn}
\noindent 
This is useful because if every partition-type of a given global type is higher than the corresponding partition-type of a second global type, then the first global type is higher in the preorder over global types. 
Hence, it is useful to systematically characterize which partition-types encode which others. We summarize the known (and unknown) encodings in Table~\ref{partitiontypes}.
\begin{figure}[htb!]
\centering
\includegraphics[width=0.48\textwidth]{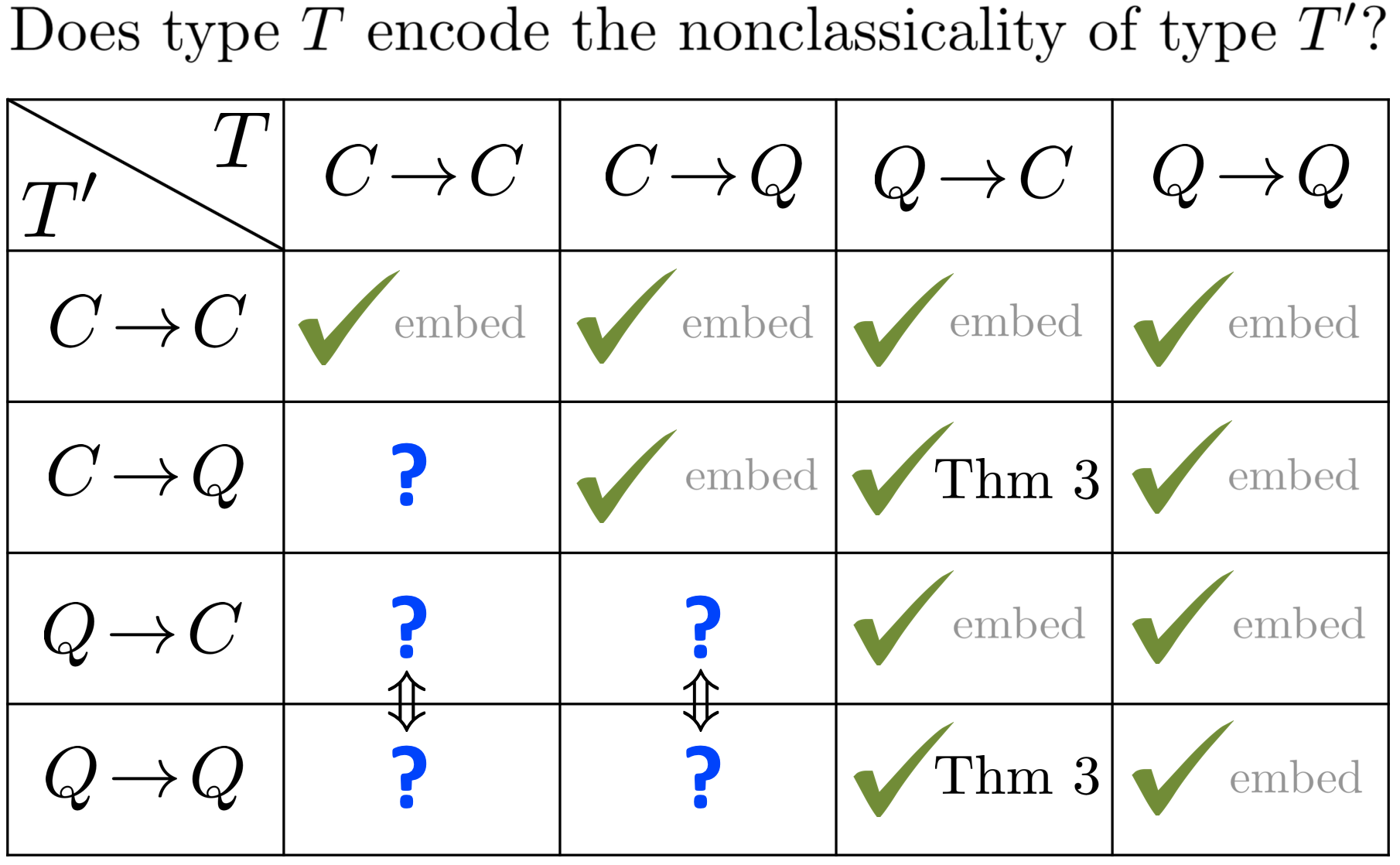}
\caption{ 
A green check mark indicates that the column partition-type $T$ is higher in the order than the row partition-type $T'$, that is, $T \succeq_{\rm type} T'$. 
The question marks indicate open questions. In two cases (indicated by $\Leftrightarrow$), a pair of these must have the same answer.} \label{partitiontypes}
\end{figure}

Clearly, partition-type $\mathsf{Q} \too \mathsf{Q}$ is at the top of the preorder over partition-types, since every other partition-type is a special case of this partition-type. We refer to such trivial encodings as {\em embeddings} of one type into another.

An immediate consequence of Theorem~\ref{trivialtypes} is that the partition-types $\mathsf{I} \too \mathsf{I}$, $\mathsf{I} \too \mathsf{C}$, $\mathsf{I} \too \mathsf{Q}$, $\mathsf{C} \too \mathsf{I}$, and $\mathsf{Q} \too \mathsf{I}$ are all at the bottom of the preorder over partition-types.

Noting that the proof of Theorem~\ref{squniv} required only local operations on each individual party, it follows immediately that the distributed measurement (or semiquantum) partition-type is at the top of the preorder. Since partition-types $\mathsf{Q} \too \mathsf{Q}$ and $\mathsf{Q} \too \mathsf{C}$ encode each other, they are in the same equivalence class at the top of the preorder.


There remain several open questions, as indicated in the table by question marks.

First, consider the questions of whether partition-type $\mathsf{C} \too \mathsf{C}$ encodes either $\mathsf{Q} \too \mathsf{C}$ or $\mathsf{Q} \too \mathsf{Q}$. Recalling that $\mathsf{Q} \too \mathsf{C}$ and $\mathsf{Q} \too \mathsf{Q}$ are in the same equivalence class, the answer to these two questions must be the same.
This question has deep implications for the task of characterizing resources in practice, since a positive answer would imply the possibility of device-independent~\cite{scarani} characterizations of the postquantumness of arbitrary resources.

Second, consider the questions of whether partition-type $\mathsf{C} \too \mathsf{Q}$ encodes either $\mathsf{Q} \too \mathsf{C}$ or $\mathsf{Q} \too \mathsf{Q}$. For the reason given just above, the answer to these two questions must again be the same.
A positive answer to this question would imply the possibility of {\em preparation}-device-independent characterizations of the postquantumness of arbitrary resources.

Finally, consider the question of whether partition-type $\mathsf{C} \too \mathsf{C}$ encodes partition-type $\mathsf{C} \too \mathsf{Q}$. A positive answer to this question would imply the possibility of device-independent characterizations of the postquantumness of any resources whose inputs are all classical.
It is also relevant for the question of whether or not postquantum Bob-with-input steering assemblages constitute a novel form of postquantumness, a question first raised in Refs.~\cite{Hoban_2018,BobWI}. 
There, the authors argued that some Bob-with-input assemblages are indeed novel, in the sense that they cannot be converted into any postquantum box by applying sets of quantum measurements to the states in the Bob-with-input assemblage. 

In light of our framework, this argument is not conclusive. Our framework provides a more principled notion of what it means for a given type $T$ to be `truly novel' with respect to another type (in this case, box-type): there should exist at least one resource of type $T$ which is in an LOSE-equivalence class which is not shared by any box-type resource. In other words, it should be that the box-type does not encode type $T$. Refs.~\cite{Hoban_2018,BobWI} do not resolve this open question, since they do not establish that there exist postquantum Bob-with-input assemblages that cannot be converted to postquantum boxes {\em by LOSE operations}. Rather, the set of operations allowed in their proof (sets of quantum measurements on Bob's steered state) is a strict subset of the full set of LOSE operations from Bob-with-input assemblages to boxes. It is not difficult to extend their argument to show that the nonfree Bob-with-input assemblage of Eq.~\eqref{assemblageex}) cannot generate any nonfree box using local operations and shared randomness. However, it remains possible that under LOSE operations, one could generate some nonfree box. 

Indeed, we prove in Appendix~\ref{sec:examples} that the PR box is {\em at least as postquantum} as the primary example of a postquantum Bob-with-input assemblage in Ref.~\cite{BobWI}. From this fact and the point of view endorsed above, it remains an open question whether or not there exist forms of postquantum steering that are truly novel.

\begin{opq} \label{opqsteer} 
Do there exist Bob-with-input steering assemblages that constitute genuinely novel forms of postquantumness---that is, assemblages whose LOSE-equivalence class contains no postquantum boxes?
\end{opq}

\section{Distributed games }
\label{unifiedgames}

Resources of various types are often used to perform distributed information-processing tasks, such as teleportation, certified randomness generation via violation of Bell inequalities, and so on. Conversely, distributed experiments such as entanglement witnessing, nonlocal games, semiquantum games, steering experiments, and teleportation games are often used to characterize the sense in which various processes are resources. Ref.~\cite{schmid2020type} defined an abstract framework of {\bf distributed games} which subsumes all of these types of tasks as special cases. Within this framework, there is a natural set of games associated to every global type $T$.

\begin{defn} \label{defn:Tgame}
For a given global type $T$, we define a distributed $\mathbf{T}${\bf -game} as a linear map from resources of type $T$ to the real numbers. A resource of type $T$ is then naturally said to be a {\bf strategy} for a $T$-game.
\end{defn}

The referee and players are all assumed to know the complete specification of the game (that is, the types of inputs and outputs, as well as the particular linear map to the real numbers) in advance. To play the game, the players process their initial shared resources (which may be of arbitrary types) using arbitrary LOSE processes, ultimately generating a resource of the appropriate type for the game---that is, a strategy for the game. Based only on the strategy they generate, the referee assigns them a score. We depict three types of $T$-games abstractly in Fig.~\ref{games}, where a closed diagram is to be interpreted as a real number, which in this case is the score assigned to the strategy (in light grey) by the particular game (in black). Note that the processing of the player's initial resource into a strategy is not depicted here.

\begin{figure}[htb!]
\centering
\includegraphics[width=0.48\textwidth]{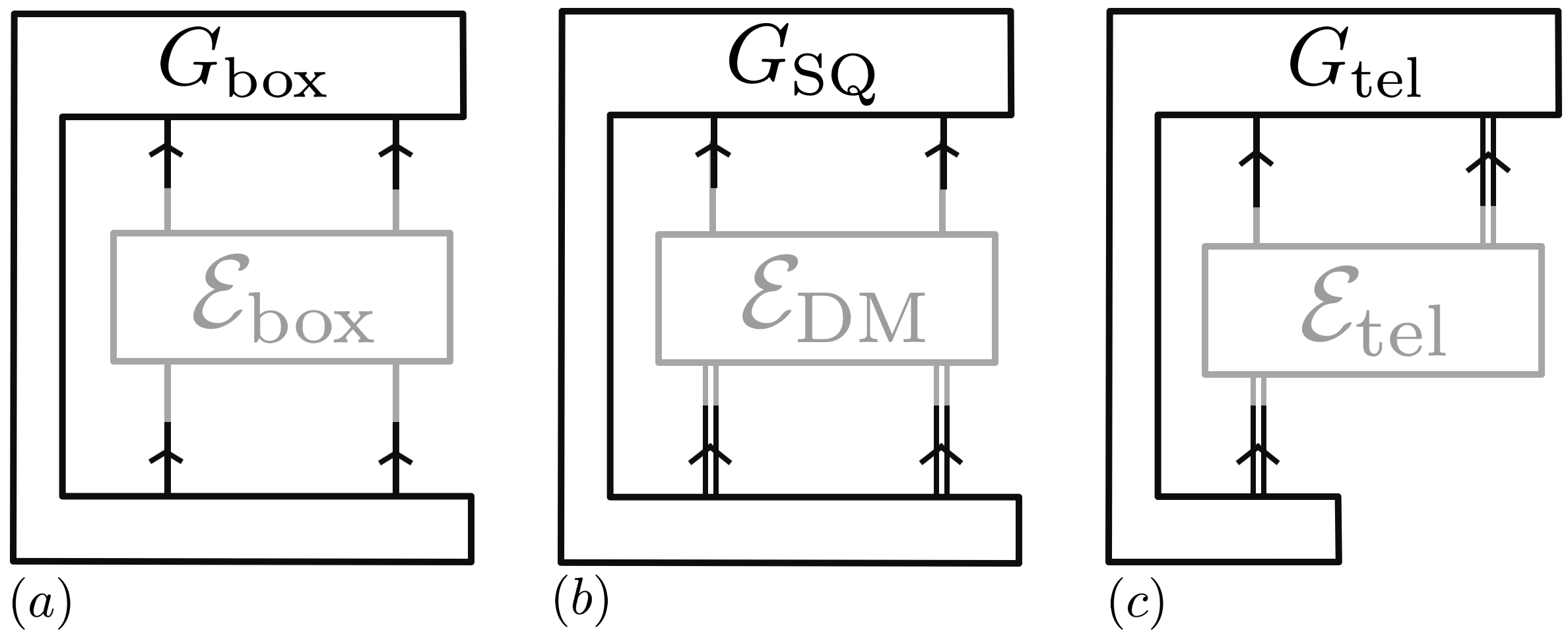}
\caption{ Some games and their strategies. (a) Boxes are strategies for nonlocal games. (b) Distributed measurements are strategies for semiquantum games. (c) Teleportages are strategies for teleportation games. As we prove, all strategies for teleportation games are LOSE-free.
} \label{games}
\end{figure}

See Ref.~\cite{schmid2020type} for a prescription by which one can experimentally implement any game of any type, as well as for a discussion of how one can interpret these games less abstractly.\footnote{In Ref.~\cite{schmid2020type}, the players are only allowed LOSR operations to generate their strategies.} Since the resources we consider are ultimately just quantum channels, one can also give a convenient mathematical decomposition of the map corresponding to any such game in terms of the trace of the given resource's Choi state with an operator in the dual space in which that Choi state lives; see e.g. Refs.~\cite{PhysRevX.8.021033,PhysRevX.9.031053,PhysRevA.101.052306}.  For our purposes, however, we merely note that the abstract definition of games above subsumes most of the usual distributed tasks, including entanglement witnesses, Bell tests, steering witnesses, teleportation games, and so on. 

Given a shared resource of any type, the distributed parties can apply arbitrary free operations to transform that resource into a strategy for the game they are playing, with the aim of optimizing their score. Denote a particular game of type $T$ by $G_{\!T}$, and the score it assigns to a strategy $\mathcal{E}$ (a specific resource of type $T$) by $G_{\!T}(\mathcal{E})$.
\begin{defn}
The {\bf optimal performance} of a resource $R$ (of arbitrary type) with respect to a game $G_{\! T}$ of arbitrary type $T$ is given by
\beq
\omega_{G_{\! T}}(R) = \max_{\tau: \mathbf{Type}[R] \! \to \! T} G_{\! T}(\tau \circ R),
\eeq
where $\tau \in \LOSE$ takes resources of the given type $\mathbf{Type}[R]$ to type $T$.
\end{defn}

Next, we define a second preorder over resources, based on their performance at games of a given type. Denote the set of all games of type $T$ by $\mathcal{G}_{\! T}$.
\begin{defn} \label{orderwrtgame}
For resources $R$ and $R'$ of different and arbitrary type, we say that $R \succeq_{\mathcal{G}_{\! T}} R'$ iff $\omega_{G_{\! T}}(R) \ge \omega_{G_{\! T}}(R')$ for every $G_{\! T} \in \mathcal{G}_{\! T}$. 
\end{defn}
\noindent This is not to be confused with the preorder over resources induced by the free LOSE-transformations, which is generally a distinct preorder.  However, Theorem~\ref{proporder} below provides a condition under which the two preorders coincide.

\subsection{Implications from the type of a resource to its performance at games} \label{implictypetoperf}

We now establish various connections between the type of a resource and how well it can perform at distributed games when the players make use of LOSE operations.

First, consider games of a given type $T$ that is LOSE-trivial. Since {\em every} resource of this type (and hence every strategy for such games) can be generated freely, one can achieve the optimal score for any such game at no cost. Hence, postquantum resources are not useful for them, and conversely, no game of type $T$ can be used to witness the postquantumness of any resource.
 The two most interesting examples of this were stated in Corollary~\ref{nopqadv}, which stated that there is no postquantum advantage for steering or teleporting quantum states, as made quantitative by steering games~\cite{schmid2020type,Cavalcanti2017} or teleportation games~\cite{schmid2020type,telep,LipkaBartosik2019}.
These two facts were pointed out in Refs.~\cite{Hoban_2018,PhysRevA.99.032334}. They follow more formally as an immediate consequence of Theorem~\ref{proporder} below.

Note that if it were the case that the task of teleportation were {\em only} concerned with establishing an identity channel between two parties who share entanglement, then it is clear that quantum theory can achieve this task perfectly, and the observation that postquantum common causes give no advantage would be trivial. However, this observation is not trivial, since many teleportation tasks are not of this form, as pointed out in Refs.~\cite{telep,LipkaBartosik2019}.

On the other hand, any nonsignaling resource of any type $T$ which is nonfree in our resource theory provides a postquantum advantage for at least one game of type $T$. 
This follows as an immediate consequence of the following (more general) theorem.
\begin{thm} \label{proporder}
If $T \succeq_{\rm type} T'$, then for resources $R_1,R_2$ of type $T'$, $R_1 \succeq_{\rm LOSE} R_2$ iff $R_1 \succeq_{\mathcal{G}_{\! T}} R_2$. 
\end{thm}
\noindent The proof is given in Appendix~\ref{neccond}. 

In other words, if type $T$ is above type $T'$, then for resources of type $T'$, the preorder over resources defined by LOSE-conversion coincides with the preorder over resources defined by performance at all games of type $T$. This implies that games of a higher type perfectly characterize the postquantumness of resources of a lower type. Conversely, it implies that every nonfree resource is useful for some game of the same type, and also for some game of every higher type.


Combining Theorem~\ref{proporder} with Theorem~\ref{trivialtypes}, it follows that there {\em are} postquantum advantages for nonlocal games, semiquantum games, measurement-device-independent steering games, channel steering games, Bob-with-input steering games, and so on. To our knowledge, these advantages have only been studied in the context of nonlocal games~\cite{Bellreview} and Bob-with-input-steering~\cite{BobWI}, while the rest remain to be studied. 

Combining Theorem~\ref{proporder} with Theorem~\ref{squniv}, it follows that the set $\mathcal{G}_{\rm SQ}$ of all semiquantum games witnesses (and exactly quantifies) the resourcefulness of any nonfree resource.
\begin{cor} \label{Buscgen}
For any resources $R$ and $R'$ (which may be of arbitrary and different types), $R \succeq_{\rm LOSE} R'$ if and only if $R \succeq_{\mathcal{G}_{\rm SQ}}R'$.  
\end{cor}

We close by noting a useful result (e.g., it is used in the Appendix for proving Theorem~\ref{proporder}), which states that if one resource outperforms a second at all possible games of a given type, then it can also generate {\em any specific strategy} of that type which the second resource can generate. 
\begin{thm} \label{subsumestrat}
For resources $R$ and $R'$ of arbitrary types and a resource $\mathcal{E}_{T}$ of arbitrary type $T$, 
$R \succeq_{\mathcal{G}_{\! T}} R'$ iff $R' \LOSEconv \mathcal{E}_{T} \implies R \LOSEconv \mathcal{E}_{T}$, that is, iff any strategy $\mathcal{E}_{T}$ for games of type $T$ that can be freely generated from $R'$ can {\em also} be freely generated from $R$. 
\end{thm}

\section{Open questions and Conclusion} \label{secopq}

We have presented the type-independent resource theory of local operations and shared entanglement, or LOSE operations. We showed how this unifies various types of postquantum resources, and discussed how these resources can be transformed into each other, regardless of their types.
This allows us to rigorously quantify the amount of postquantumness in a given resource, and to quantitatively compare the postquantumness of arbitrary resources of arbitrary types. Next, we began the systematic characterization of which types of resources are able to express all the same manifestations of postquantumness as which other types, as made formal by the idea of encodings of types of resources. We then discussed how players can use any given resource to generate strategies for distributed games of arbitrary types, and how the type of a given resource helps determine what advantage, if any, the players can hope to gain over strategies that are achievable by quantum common causes.

There remain a great deal of interesting questions to be understood about this resource theory; we have already highlighted some of these in the main text.
Attaining a complete understanding of postquantum common cause resources will require answering all of the usual resource theoretic questions: what is the structure of the single-copy deterministic preorder, what monotones and witnesses provide valuable information about this preorder, what asymptotic and catalytic conversions are possible, and so on.
A modest starting point would be to finish characterizing which conversions between known postquantum resources (such as those in Appendix~\ref{sec:examples}) are possible, and to find type-independent monotones that witness the interconversions that are not possible. Note that for the type-independent LOSR resource theory, Ref.~\cite{rosset2019characterizing} developed technical tools for answering such questions, and it seems likely that many of these techniques (such as the robustness monotones defined therein) would be transferrable with minor modifications to the LOSE resource theory. It would also be interesting to study detailed properties of the LOSE preorder even for specific types of resources, in analogy to what was done in Ref.~\cite{wolfe2020quantifying} for box-type resources in the resource theory of LOSR. It would also be helpful to find some algorithmic methods for determining if a given channel is LOSE-free or not; see Appendix~\ref{neccond} for three partial results along these lines from Ref.~\cite{causallocaliz}. 
See also Appendix~\ref{bennettq} for an example of a well-known channel which we conjecture is LOSE-nonfree. Another interesting question is whether the generalized notion of self-testing of arbitrary types of resources introduced in Ref.~\cite{LOSRvsLOCCentang} could be applied in this resource theory.
Finally, it is an interesting question whether or not resource theories like ours will be of use at constraining the set of quantum correlations (or more general processes) from physical principles.
\section{Acknowledgments}
D.S. is supported by a Vanier Canada Graduate Scholarship.
Research at Perimeter Institute is supported in part by the Government of Canada through the Department of Innovation, Science and Economic Development Canada and by the Province of Ontario through the Ministry of Colleges and Universities. MJH acknowledges the FQXi large grant The Emergence of Agents from Causal Order. This publication was made possible through the support of a grant from the John Templeton Foundation. The opinions expressed in this publication are those of the authors and do not necessarily reflect the views of the John Templeton Foundation.

\setlength{\bibsep}{2pt plus 1pt minus 2pt}
\bibliographystyle{apsrev4-1}
\nocite{apsrev41Control}
\bibliography{bib}

\appendix

\section{Explicit examples of nonfree resources} \label{sec:examples}\


The majority of this article focuses on what can be said about resources by considering only their type, and hence does not focus on particular instances of resources. Nonetheless, it is illustrative to see some concrete examples of postquantum resources. We now describe five explicit examples of four different types. We denote Alice's input (output) system by $A$ ($X$) and Bob's input (output) system by $B$ ($Y$). If any of these systems are classical, then we will identify the corresponding capital letter with the classical variable corresponding to that system, and will denote the possible values that it can take by $a,x,b,y$, respectively, where $a,x,b,y \in \{0,1\}$ unless otherwise specified. We will also denote the three Pauli matrices by $\sigma_1,\sigma_2,$ and $\sigma_3$.

\subsection{The PR box}
Our first explicit example is the celebrated Popescu-Rohrlich (PR) box~\cite{Popescu1994}, defined as the resource that achieves the logically maximal violation of the (CHSH) inequality~\cite{CHSH}. The PR box is a resource of type $\mathsf{CC} \too \mathsf{CC}$, and is completely specified by the following conditional probability distribution:
\begin{equation} 
    \mbox{p}(ab|xy) =
    \begin{cases}
        1/2 \quad &\text{if } a \oplus b = xy, \\
        0 \quad &\text{otherwise}. 
    \end{cases}
\end{equation} 
The PR box is well-known precisely because it cannot be generated by local operations and shared entanglement. More specifically, it achieves a value of 4 for the CHSH functional, while the largest value achievable by local measurements on a shared entangled state is $2 \sqrt{2}$.

\subsection{The PHHH ensemble-preparation} ~\label{ensembex}
Our second example of a postquantum resource is an ensemble-preparing channel introduced in Ref.~\cite{nosigboxes}, which we refer to as the PHHH ensemble-preparation.\footnote{ Ref.~\cite{nosigboxes} presented this resource as one with quantum inputs, but since these are immediately dephased in the computational basis, one can simply treat them as classical. Note also that an incoherent version of this channel was shown to be postquantum in Ref.~\cite{causallocaliz}.}
This resource has type $\mathsf{CC} \too \mathsf{QQ}$, and is constructed as follows. If the product of the inputs $xy=0$, then the channel outputs the Bell state $\ket{\phi^+} = (\ket{00}+\ket{11})/\sqrt2$, but if $xy = 1$ it outputs the Bell state $\ket{\psi^+} = (\ket{01}+\ket{10})/\sqrt2$. 

A quantum circuit that realizes the PHHH ensemble-preparation using classical communication from Bob to Alice is shown in Fig.~\ref{PHHH}. 

\begin{figure}[htb!]
\centering
\includegraphics[width=0.25\textwidth]{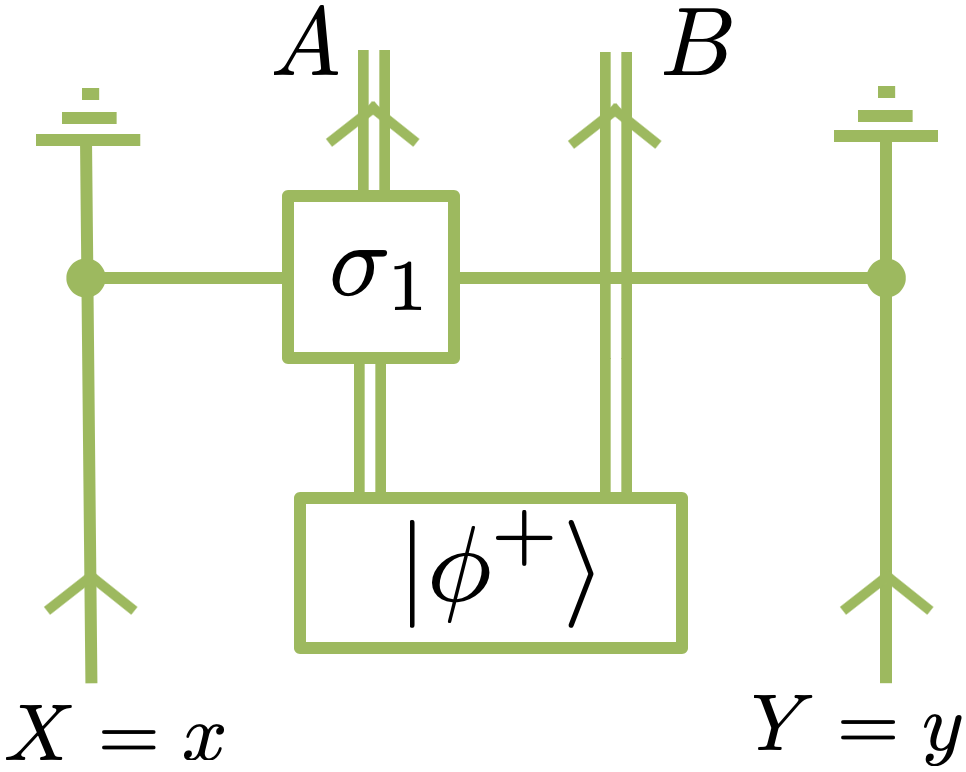}
\caption{ 
A quantum circuit that realizes the PHHH ensemble-preparation. Note that this circuit requires communication from Bob to Alice, and hence does not constitute an LOSE realization of the channel.
 Here and throughout, controlled operations are denoted by a horizontal line ending in a dot on the target system. 
} \label{PHHH}
\end{figure}

{\bf Converting the PR box into the PHHH ensemble-preparation and back---}
To see that the PHHH ensemble-preparation is LOSE-nonfree, we show that it is in the same equivalence class as the PR box; that is, the two are interconvertible using LOSE operations. 

To obtain a PR box from the PHHH ensemble-preparation, it suffices for each party to locally dephase their output system in the computational basis, resulting in a $\mathsf{CC} \too \mathsf{CC}$ channel with exactly the PR box correlations in the computation basis. To obtain the PHHH ensemble-preparation from the PR box, the two parties begin with maximally entangled state $\ket{\phi^+}$, and each party performs a local Pauli $\sigma_1$ transformation if the output they receive from the PR box is $1$. When $xy = 0$, the PR box outputs perfectly correlated bits, and hence either both parties perform the transformation, or neither party does; in either case, the state $\ket{\phi^+}$ is invariant. When $xy = 1$, however, the PR box outputs perfectly anticorrelated bits, and so only one party performs the Pauli $\sigma_1$ transformation; in either case, the resulting state is $\ket{\psi^+}$. The overall channel,
depicted in Fig.~\ref{Piani},
is exactly the PHHH ensemble-preparation.

\begin{figure}[htb!]
\centering
\includegraphics[width=0.3\textwidth]{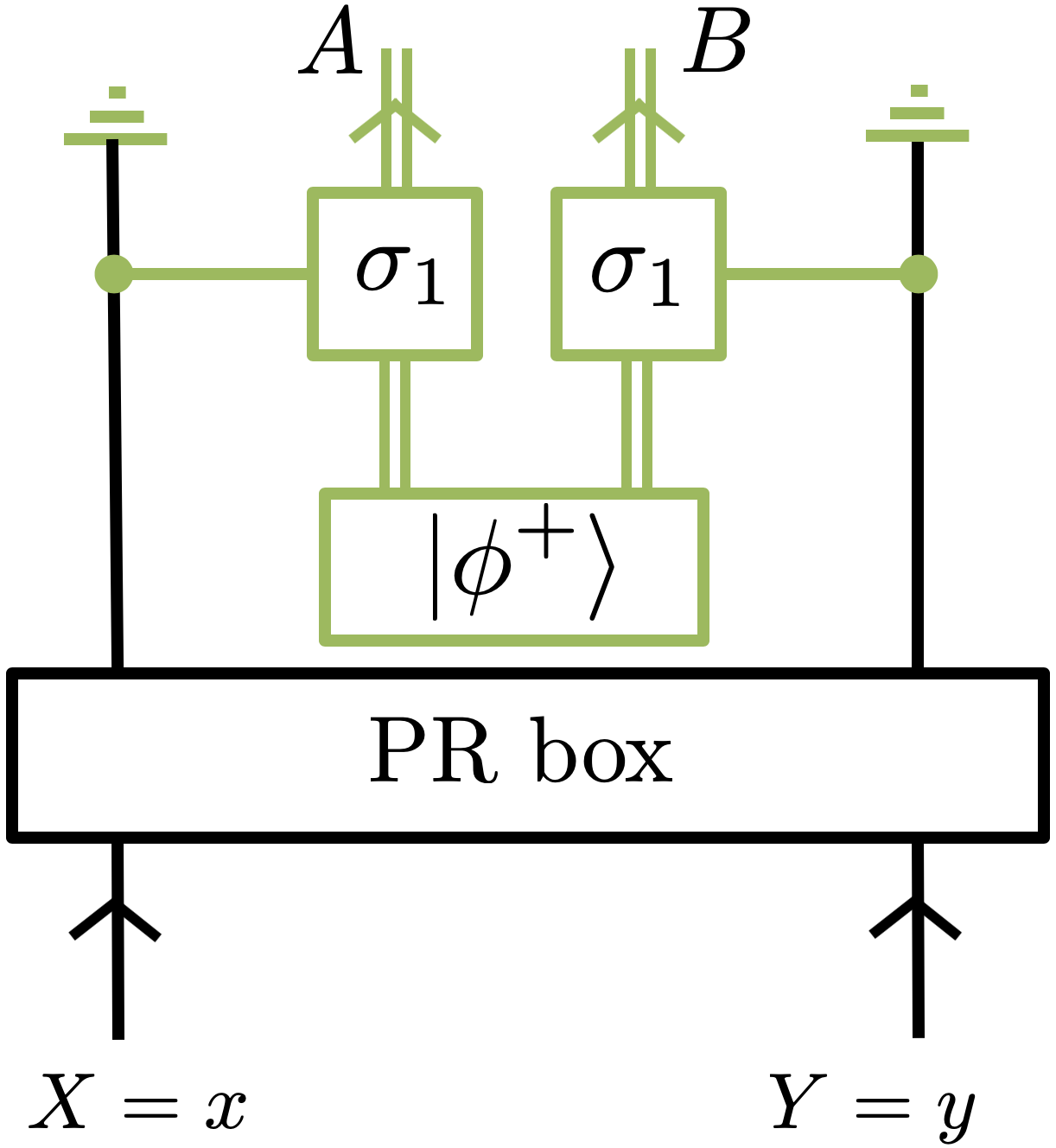}
\caption{ 
Generating the PHHH ensemble-preparation from the PR box using LOSE operations (green). 
} \label{Piani}
\end{figure}

\subsection{The SHSA Bob-with-input steering assemblage} ~\label{BWIex}
Our third example is a postquantum Bob-with-input (BWI) steering assemblage, of type $\mathsf{CC} \too \mathsf{CQ}$, that was introduced in Ref.~\cite{BobWI}. We refer to it as the SHSA BWI-assemblage. A Bob-with-input steering scenario generalizes the traditional steering scenario by allowing Bob a classical (rather than trivial) input. 
A BWI-assemblage is fully specified
by the set  $\{\rho_{a|xy}\}_{a,x,y}$ of unnormalized quantum states (on system $B$) that Bob can be steered into, indexed by the values $X=x$ ($A=a$) of Alice's classical input (output) and the value $Y=y$ of Bob's classical input. 

The SHSA BWI-assemblage is defined in Ref.~\cite{BobWI} by
\begin{equation} \label{assemblageex}
    \rho_{a|xy} = \frac{1}{4} \left( I + (-1)^a \sigma_{x+1} \right)^{T^y}.
\end{equation}
Here, $\sigma_x$ are the three Pauli matrices, depending on the value $x \in \{0,1,2\}$, and $T^y$ is the transpose (in the computational basis) on Bob's state when $y=1$ and the identity map when $y=0$. 

Mathematically (but {\em not} physically), this assemblage can be depicted as in Fig.~\ref{figassemblageex}, where a measurement of the Pauli observable $(-1)^{x}\sigma_{x+1}$ (depending on the input value $x$, and with outcomes $0$ and $1$ corresponding to eigenvalues $+1$ and $-1$, respectively) is performed on half of a Bell state $\ket{\phi^+}$, and $T^y$ is applied on the other half, depending on the input value $y$. The set of states on $B$ conditioned on particular values $X=x$, $A=a$, and $Y=y$ is given by Eq.~\eqref{assemblageex}. Crucially, the transpose is not a valid quantum channel, since it is not completely positive. Hence, Fig.~\ref{figassemblageex} does not represent a valid quantum circuit, much less a set of LOSE operations.

\begin{figure}[htb!]
\centering
\includegraphics[width=0.45\textwidth]{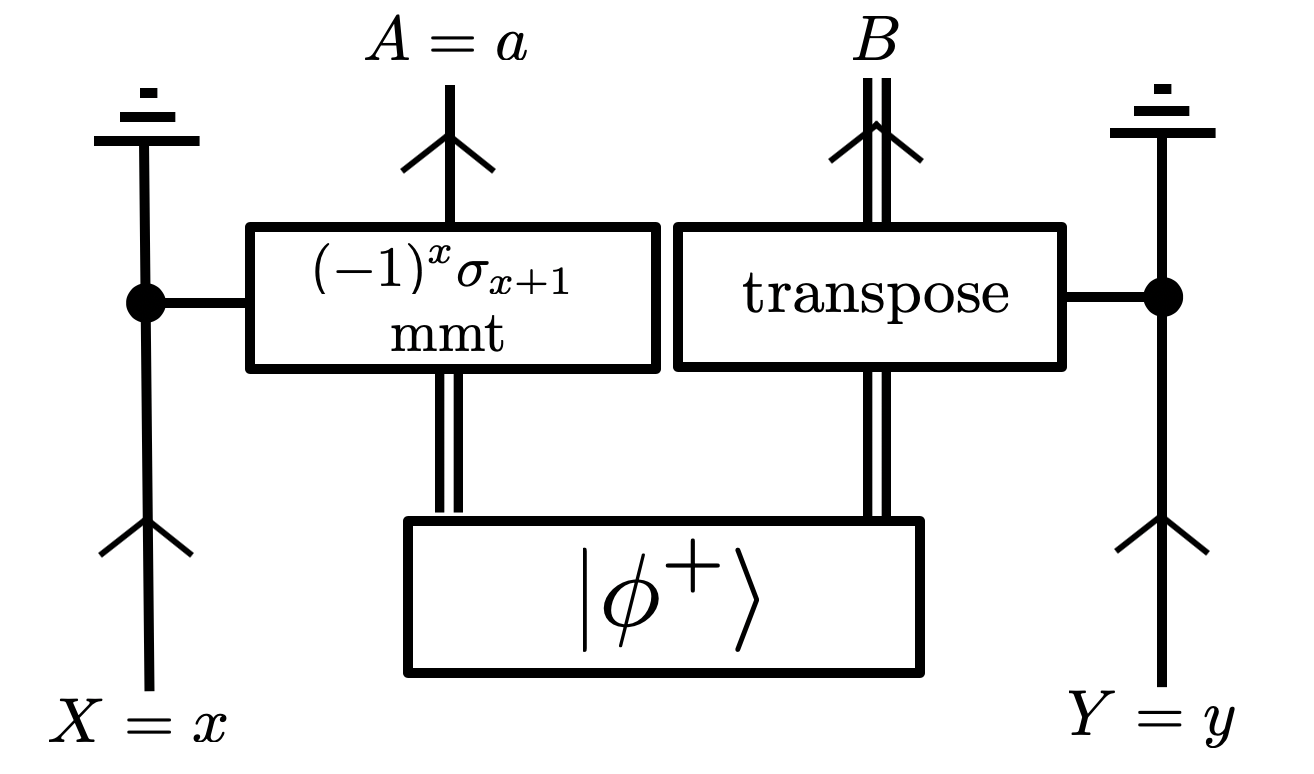}
\caption{ 
An (unphysical) depiction of the assemblage of Eq.~\eqref{assemblageex}. Note that this is not a quantum circuit, much less an LOSE process, since the transpose is not a CPTP map.
} \label{figassemblageex}
\end{figure}

Fig.~\ref{figassemblageex} provides some intuition for why the SHSA BWI-assemblage
is nonfree, as the transpose map is not physical. The formal proof that it is nonfree is more elaborate, and is given in the supplementary material of Ref.~\cite{BobWI}.

{\bf Converting the PHHH ensemble-preparation into the BWI-assemblage---}
Next, we prove that the PHHH ensemble-preparation can be converted into the SHSA BWI-assemblage using LOSE operations. The conversion is shown in Fig.~\ref{PRtoBWI}. By transitivity, this also proves that the PR box can also be converted into the SHSA BWI-assemblage.

\begin{figure}[htb!]
\centering
\includegraphics[width=0.3\textwidth]{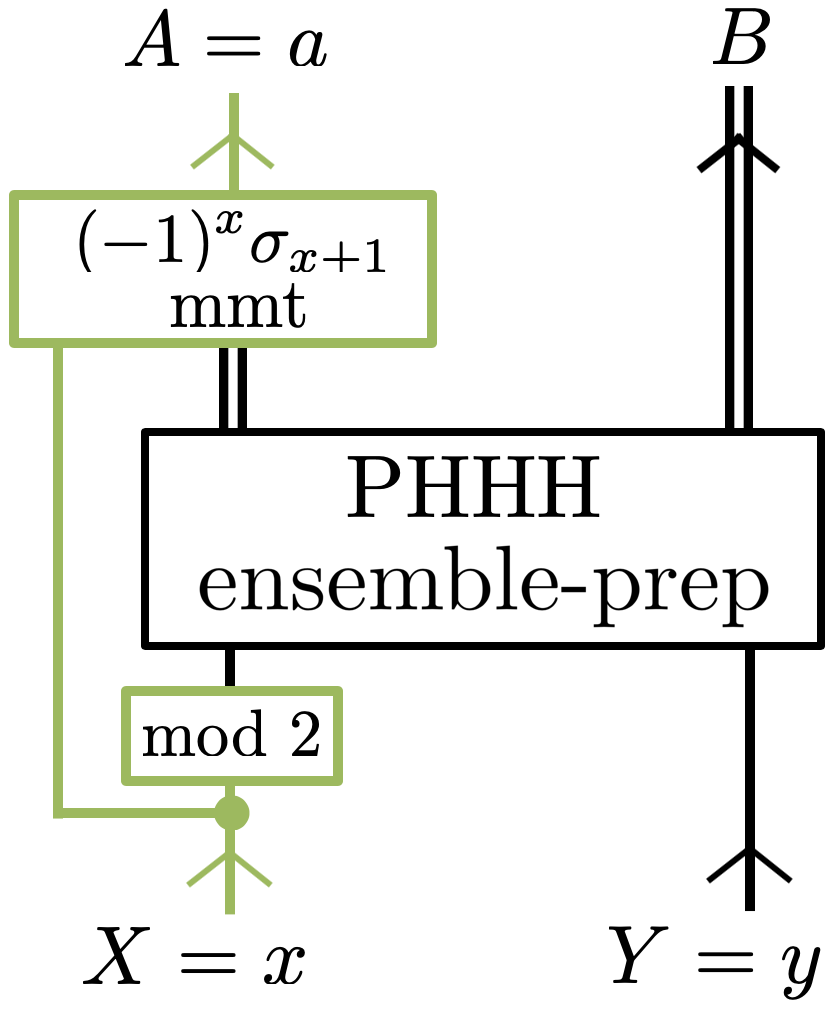}
\caption{ 
Generating the SHSA BWI-assemblage from the PHHH ensemble-preparation using LOSE operations (green).} \label{PRtoBWI}
\end{figure}

First, Alice maps her given input $x\in \{0,1,2\}$ to $x {\rm \ mod \ } 2$ and uses this as her input to the PHHH ensemble-preparation (but keeps a copy of $x$ as an input to her post-processings). On her quantum output, Alice then measures the Pauli observable $(-1)^{x}\sigma_{x+1}$, depending on her input value $x$, where the outcomes $0$ and $1$ correspond to eigenvalues $+1$ and $-1$, respectively. 
It is easy to verify that this reproduces the SHSA BWI-assemblage, e.g. by considering the 12 possible states generated for the 12 possible tuples $(a,x,y)$.

\subsection{The BGNP channel} \label{pqmmttwisted}
Our fourth example is a generic nonsignaling channel proposed in \cite{causallocaliz}, which we term the BGNP channel.\footnote{Strictly speaking, it is a family of channels indexed by a unitary $U_B$.}
It is a resource of type $\mathsf{QQ} \too \mathsf{QQ}$, where every input and output has dimension 4.
Let $\{ \ket{0}_A, \ket{1}_A, \ket{2}_A, \ket{3}_A \}$ be a basis for Alice's Hilbert space, and let $\{ \ket{0}_B, \ket{1}_B, \ket{2}_B, \ket{3}_B \}$ be a basis for Bob's Hilbert space. 
For Alice as well as for Bob, we refer to the subspace spanned by $\{\ket{0},\ket{1}\}$ as the first qubit, labeled $f$, and the subspace spanned by $\{\ket{2},\ket{3}\}$ as the second qubit, labeled $s$. 
One can form a basis for the 16-dimensional joint Hilbert space by considering the 4 Bell states on each of the two-qubit subspaces formed by taking one of Alice's qubits ($f$ or $s$) together with one of Bob's qubits ($f$ or $s$); explicitly, this is the basis
\begin{equation}
    \Big\{ \ket{\phi_{ff}^\pm}\! ,\ket{\psi_{ff}^\pm}\! , \ket{\phi_{fs}^\pm}\! ,\ket{\psi_{fs}^\pm}\! ,\ket{\phi_{sf}^\pm}\! ,\ket{\psi_{sf}^\pm}\! , \ket{\phi_{ss}^\pm}\! ,\ket{\psi_{ss}^\pm} \Big\}.
\end{equation}
It follows that
\begin{align} \label{twistedbasis}
    \Big\{ \ket{\phi_{ff}^\pm}\! ,&\ket{\psi_{ff}^\pm}\! , \ket{\phi_{fs}^\pm}\! ,\ket{\psi_{fs}^\pm}\! ,\ket{\phi_{sf}^\pm}\! ,\ket{\psi_{sf}^\pm}\! ,
    \\ \nonumber
   & (I_A \otimes U_B)\ket{\phi_{ss}^\pm}\! , (I_A \otimes U_B)\ket{\psi_{ss}^\pm} \Big\}
\end{align}
is also a basis, since it amounts to merely rotating the Bell basis of the $ss$ subspace by a unitary $U_B$ acting (only) on Bob's second qubit. This `twisted' basis is depicted schematically in Fig.~\ref{twisted}.

\begin{figure}[htb!]
\centering
\includegraphics[width=0.4\textwidth]{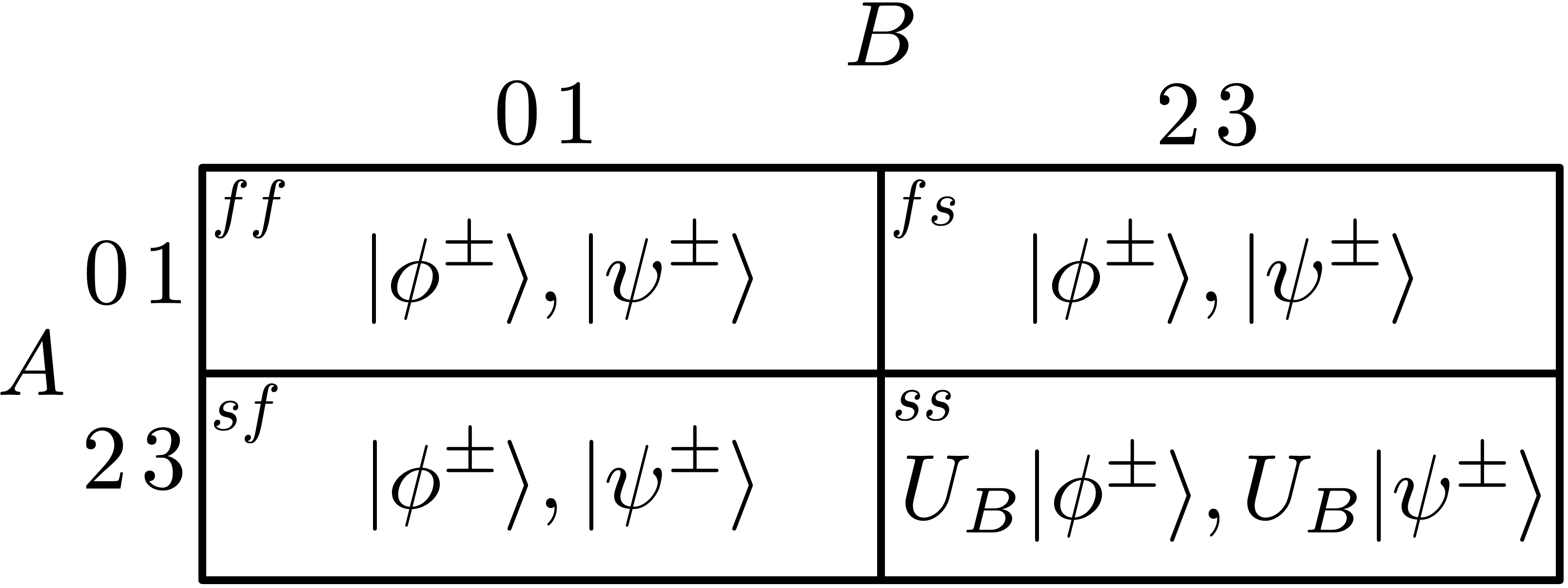}
\caption{ 
The basis of 16 states defined by Eq.~\eqref{twistedbasis}.
} \label{twisted}
\end{figure}

The BGNP channel is defined as
the channel which is fully dephasing in this latter basis (that is, as the channel whose Kraus operators are the projectors onto the states in this basis). It is proved in Ref.~\cite{causallocaliz} that this channel cannot be performed using LOSE operations. The proof relies on a particular necessary condition for a channel to be LOSE-free (which we repeat in Appendix~\ref{neccond}) which is not satisfied by the channel.

Ref.~\cite{causallocaliz} also provides an intuition for why this channel cannot be achieved via LOSE operations. The authors note that the channel which is dephasing in the Bell basis for two qubits can be performed using only local operations and shared randomness, as can the Bell measurement rotated by $U_B$. However, to implement the measurement defined by Eq.~\eqref{twistedbasis}, the parties would have to know whether or not the state lies in the $ss$ subspace (to determine which of these two Bell measurements to perform), and neither Alice nor Bob has access to this information in the absence of communication.

As an aside: this channel is defined by a projective measurement, and it always outputs a quantum state which is diagonal in a fixed basis. As such, one might expect that it can be viewed as a distributed measurement, that is, a resource of type $\mathsf{QQ} \too \mathsf{CC}$. However, it is not natively a resource of that type. This is because the output state is not diagonal {\em in the computational basis}, nor even in a product basis from which Alice and Bob can read off a classical outcome locally.

\subsection{The DFP channel}

All of the previous examples are entanglement-breaking. 
Our fifth example, which we refer to as the DFP channel, is a postquantum channel of type $\mathsf{QQ} \too \mathsf{QQ}$ that was introduced in Ref.~\cite{PerinottiLOSEex} as an example of an LOSE-nonfree resource which is {\em not} entanglement-breaking. (Note that {\em all} bipartite resources of types other than $\mathsf{QQ} \too \mathsf{QQ}$ are necessarily entanglement-breaking.)  The example is essentially a coherent mixture of a bipartite identity channel with the PHHH ensemble-preparation defined above (but where the bipartite quantum state used to implement the coherent mixture is also an output of the channel). 

Explicitly, the DFP channel can be defined by the circuit in Fig.~\ref{Paolo} for $\alpha = \frac{1}{6}$. 
\begin{figure}[htb!]
\centering
\includegraphics[width=0.35\textwidth]{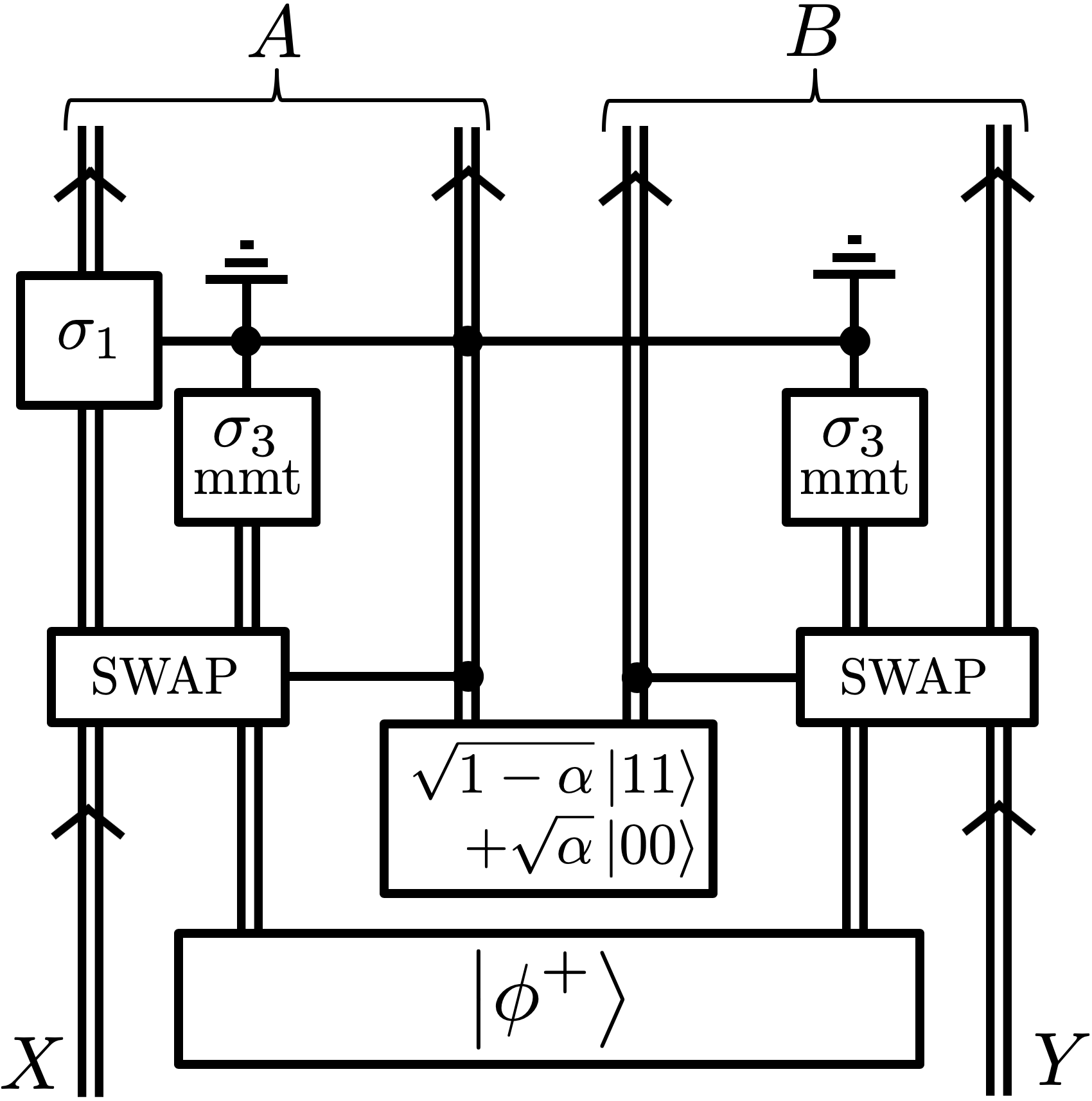}
\caption{ 
A quantum circuit that realizes the DFP channel when $\alpha = \frac{1}{6}$. Note that this circuit requires communication from Bob to Alice, and hence does not constitute an LOSE realization of the channel.
} \label{Paolo}
\end{figure}
Here, the Pauli unitary $\sigma_1$ is controlled on three systems; namely, it is only applied if the outcomes of both computational basis measurements (labeled by `$\sigma_3$ mmt') are $1$ and also if the quantum system on which it is controlled is in state $\ket{1}$. More details on this channel, including its Kraus representation, can be found in Ref.~\cite{PerinottiLOSEex}.

For $\alpha = \frac{1}{6}$, Ref.~\cite{PerinottiLOSEex} proves that this channel can be freely converted into a box-type resource that violates the CHSH inequality beyond the quantum bound of $2\sqrt{2}$, and hence that it is LOSE-nonfree, and also proves that its Choi state has a partial transpose with negative eigenvalues, and hence that it is not entanglement-breaking.  Intuitively, the fact that the channel is LOSE-nonfree is inherited from the component of the PHHH ensemble-preparation, while the fact that the channel is not entanglement-breaking is inherited from the component of the identity channel. They also prove that it is extremal in the space of channels.

{\bf Converting the PR box to the DFP channel}

Fig.~\ref{PHHHtoDFP} presents a construction for converting the PHHH ensemble-preparation into the DFP channel using LOSE operations. 

\begin{figure}[htb!]
\centering
\includegraphics[width=0.475\textwidth]{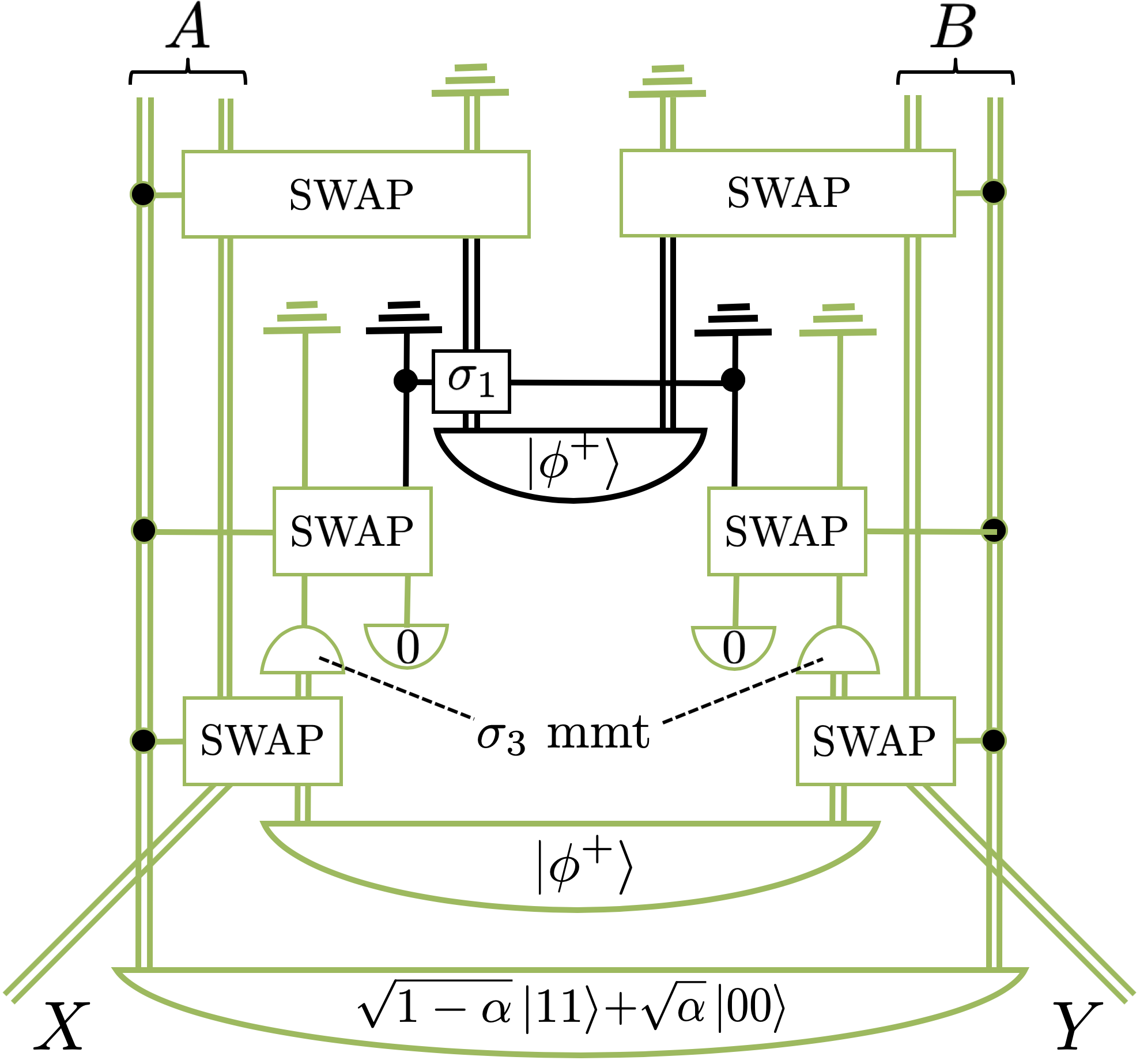}
\caption{ 
Generating the DFP channel from the PHHH ensemble-preparation (in black) using LOSE operations (in green).
} \label{PHHHtoDFP}
\end{figure}

Here, the process labeled $0$ represents a fixed preparation of state $\ket{0}$. That the overall channel in Fig.~\ref{PHHHtoDFP} is equivalent to the DFP channel in Fig.~\ref{Paolo} is straightforward but tedious to check by linear algebra. Intuitively, one can see that the construction works because the controlled swap operations ensure that either both Alice and Bob apply the identity transformation to their inputs, or they measure them in the $\sigma_3$ basis and implement the PHHH ensemble-preparation, controlled on the state of their quantum system being $\ket{0}$ or $\ket{1}$, respectively.

\subsection{The Bennett channel} \label{bennettq}
Our final example is a resource that we suspect is postquantum, though we have not yet been able to prove this. 
The resource is defined by dephasing in the well-known product basis considered in Ref.~\cite{Bennett}. We will refer to it as the Bennett channel. It is a resource of type $\mathsf{QQ} \too \mathsf{QQ}$, where every input and output has dimension 3.

Ref.~\cite{Bennett} considers a particular basis of product states for two qutrits. With the definitions
\begin{align}\label{eq:bennett}
    &\ket{\alpha_1} = \ket{1}& &\ket{\beta_1} = \ket{1}, \nonumber \\
    &\ket{\alpha_2} = \ket{0}& &\ket{\beta_2} = \ket{0} + \ket{1} \nonumber \\
    &\ket{\alpha_3} = \ket{0}& &\ket{\beta_3} = \ket{0} - \ket{1} \nonumber \\
    &\ket{\alpha_4} = \ket{2}& &\ket{\beta_4} = \ket{1} + \ket{2} \nonumber \\
    &\ket{\alpha_5} = \ket{2}& &\ket{\beta_5} = \ket{1} - \ket{2}  \\
    &\ket{\alpha_6} = \ket{1} + \ket{2}& &\ket{\beta_6} = \ket{0} \nonumber \\
    &\ket{\alpha_7} = \ket{1} - \ket{2}& &\ket{\beta_7} = \ket{0} \nonumber \\
    &\ket{\alpha_8} = \ket{0} + \ket{1}& &\ket{\beta_8} = \ket{2} \nonumber \\
    &\ket{\alpha_9} = \ket{0} - \ket{1}& &\ket{\beta_9} = \ket{2} \nonumber 
\end{align}
then  the relevant basis for the 9-dimensional input space defined by the two qubits is given (up to normalization) by 
\begin{equation} \label{bennett1}
\Big\{ \ket{\alpha_i} \otimes \ket{\beta_i} \Big\}_{i = 1,2,...,9}.
\end{equation}
The 9 states are depicted in Fig.~\ref{Bennett}.

\begin{figure}[htb!]
\centering
\includegraphics[width=0.3\textwidth]{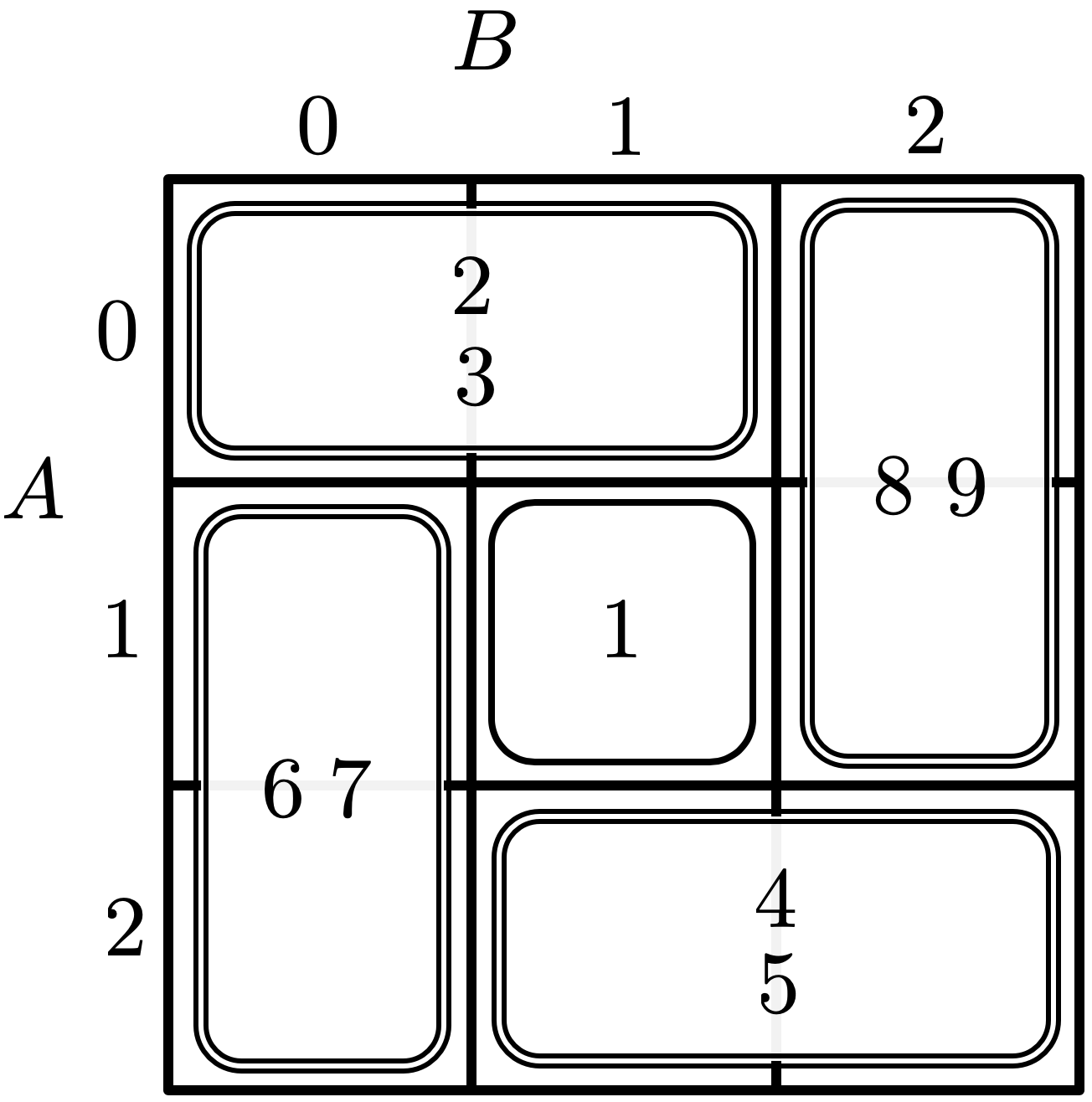}
\caption{ 
The basis of 9 states defined in Eq.~\eqref{bennett1} (via Eq.~\eqref{eq:bennett}).
} \label{Bennett}
\end{figure}

In \cite{Bennett}, it is shown that the channel defined by dephasing in this product basis cannot be performed using local operations and classical communication between the two parties. 
However, it is unclear whether or not this channel is a postquantum resource.
\begin{opq}
Can the Bennett channel be realized using LOSE operations?
\end{opq}



\section{Proofs of theorems from the main text} \label{sec:proofs}

The three proofs in these appendices are exactly the same as the corresponding proofs in Ref.~\cite{schmid2020type}, but with LOSR replaced by LOSE.

\subsection{Proof of Theorem~\ref{squniv}} \label{proof1}

We begin by proving Theorem~\ref{squniv}. 

\begin{proof}
Consider a bipartite channel $\mathcal{E}$ which has a quantum output of dimension $d$, together with arbitrary other outputs and inputs (denoted by dashed double lines), as shown in black in Fig.~\ref{SQandBack}(a). One can transform $\mathcal{E}$ into a resource with a quantum input of dimension $d$ and a classical output of dimension $d^2$ by composing $\mathcal{E}$ with a Bell measurement as shown in green in Fig.~\ref{SQandBack}(a); that is, by performing a measurement in a maximally entangled basis on the quantum output of $\mathcal{E}$ and a new quantum input of the same dimension $d$.
\begin{figure}[htb!]
\centering
\includegraphics[width=0.48\textwidth]{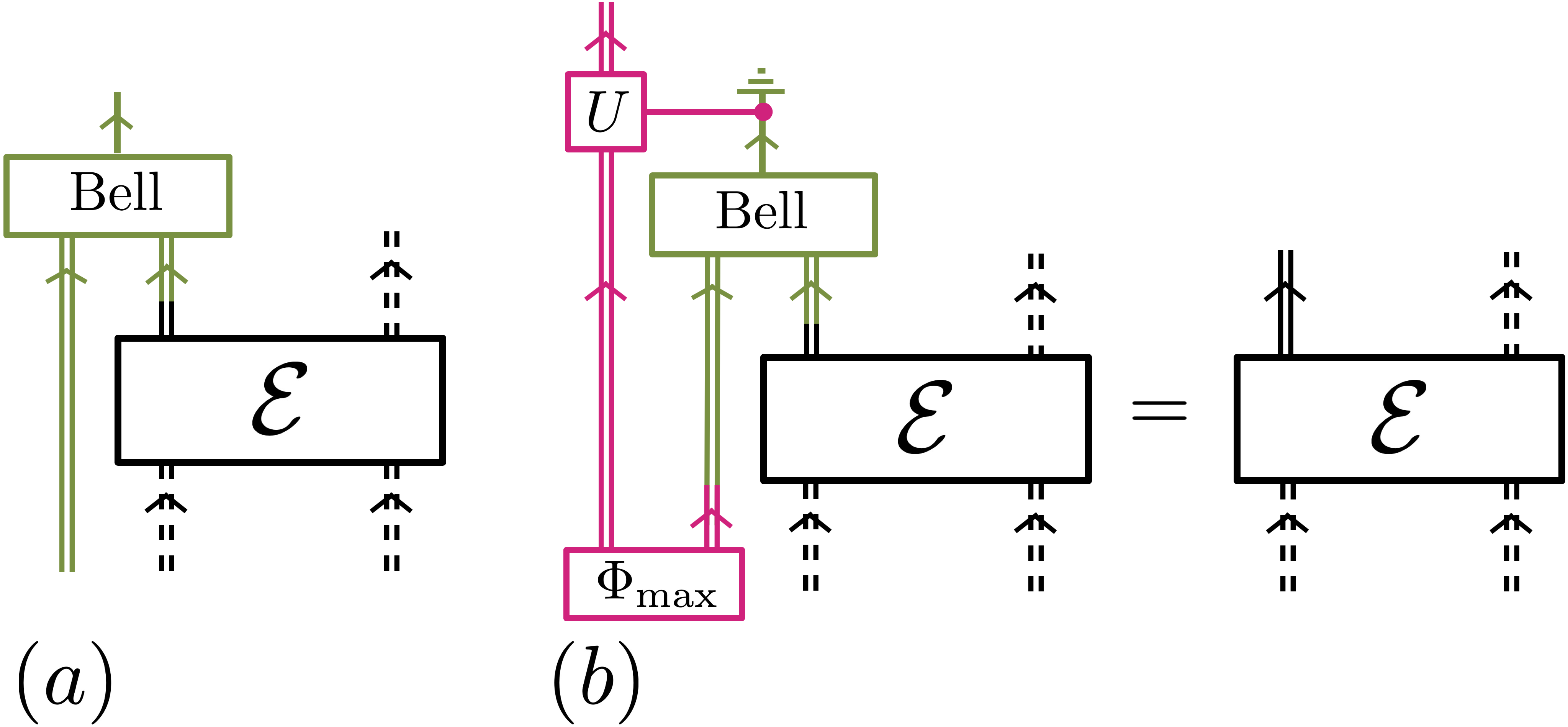}
\caption{  (a) A free transformation (in green) that converts a quantum output to a classical output together with a new quantum input. 
(b) This transformation does not change the LOSE equivalence class, since it has a left inverse (shown in pink) which is a free transformation.} \label{SQandBack}
\end{figure}
To see that this transformation preserves LOSE equivalence class, it suffices to note that there exists a local (and hence free) operation, shown in pink on the left-hand side of Fig.~\ref{SQandBack}(b), which takes the transformed channel back to the original channel $\mathcal{E}$. In particular, this local operation feeds one half of a maximally entangled state $\Phi_{\rm max}$ into the Bell measurement, and then performs a correcting unitary operation $U$ on the other half of the entangled state, conditioned on the classical outcome of the Bell measurement. For the correct choice of correction operations, the overall transformation on $\mathcal{E}$ is just the well-known teleportation protocol~\cite{Bennett93}, and so the equality shown in Fig.~\ref{SQandBack}(b) holds. Hence, the channel in Fig.~\ref{SQandBack}(a) is in the same LOSE equivalence class as $\mathcal{E}$, which implies that every partition of a resource can be transformed to a resource of type $\mathsf{Q} \too \mathsf{C}$ in the same equivalence class.
\end{proof}

\subsection{Proof of Theorem~\ref{subsumestrat}}

We now prove Theorem~\ref{subsumestrat}, since it will be useful for the next proof. 

\begin{proof} \label{proof2}

This proof refers to the decomposition of a given game into an analyzer and a payoff function, as introduced and detailed in Ref.~\cite{schmid2020type}.

If $R' \LOSEconv \mathcal{E}_{T} \implies R \LOSEconv \mathcal{E}_{T}$, then $R$ can generate any strategy for any given game $G_{\! T}$ that $R'$ can, and so always performs at least as well as $R'$ at $T$-games, and so $R \succeq_{\mathcal{G}_{\! T}} R'$.

To prove the converse, consider a set of games of type $T$ defined by ranging over all possible payoff functions $F_{\rm payoff}(abxy)$ for some fixed analyzer $Z$---that is, a specific tomographically complete measurement for each output system of the resource and a specific tomographically complete set of states for each input system of the resource.
Assume that $R' \LOSEconv \mathcal{E}_{T}$ for some strategy $\mathcal{E}_{T}$, and define $P_{Z\circ \mathcal{E}_T}(ab|xy) = Z \circ \mathcal{E}_{T}$. 
For $R \succeq_{\mathcal{G}_{\! T}} R'$, it must be that $R \LOSEconv \mathcal{E}_{T}'$ for at least one strategy $\mathcal{E}_{T}'$ satisfying $P_{Z\circ \mathcal{E}'_T}(ab|xy) = Z \circ \mathcal{E}_{T}'$. If this were {\em not} the case, then the convex set $S(R)$ of all correlations which $R$ can generate in this scenario, $S(R):=\left\{P_{Z \circ \tau \circ R}(ab|xy)=Z \circ \tau \circ R \right\}_{\tau \in {\rm LOSE}}$, would not contain $P_{Z\circ \mathcal{E}_T}(ab|xy)$, and the hyperplane which separated $P_{Z\circ \mathcal{E}_T}(ab|xy)$ from $S$ would constitute a payoff function $F_{\rm payoff}$ for which $R'$ outperformed $R$, which would be in contradiction with the claim that $R \succeq_{\mathcal{G}_{\! T}} R'$. By tomographic completeness, the preimage of every correlation  under $Z$ contains at most one strategy. Hence, if two strategies map to the same correlation, then they must be the same strategy, and so it must be that $\mathcal{E}_{T}=\mathcal{E}_{T}'$ in argument above. That is, we have shown that if $R \succeq_{\rm SQ} R'$ and $R' \LOSEconv \mathcal{E}_{T}$, then $R \LOSEconv \mathcal{E}_{T}$.
\end{proof}

\subsection{Proof of Theorem~\ref{proporder}} \label{proof3}

Finally, we prove Theorem~\ref{proporder}.

\begin{proof}
Consider the set $\mathcal{G}_{\! T}$ of all games of type $T$ and two resources $R_1$ and $R_2$, both of type $T'$, where $T \succeq_{\rm type} T'$.
Clearly $R_1 \succeq_{\rm LOSE} R_2$ implies $R_1 \succeq_{\mathcal{G}_{\! T}} R_2$, since $R_1 \succeq_{\rm LOSE} R_2$ implies that $R_1$ can be used to freely generate $R_2$ and hence to generate any strategy which can be generated using $R_2$.
Next, we prove that $R_1 \succeq_{\mathcal{G}_{\! T}} R_2$ implies $R_1 \succeq_{\rm LOSE} R_2$. By assumption, $T \succeq_{\rm type} T'$, and so for $R_2$ of type $T'$, there exists a strategy $\mathcal{E}_{ T}$ for games of type $T$ such that $R_2 \LOSEinterconv \mathcal{E}_{ T}$.
Since $R_1 \succeq_{\mathcal{G}_{\! T}} R_2$, Theorem~\ref{subsumestrat} tells us that $R_2 \LOSEconv \mathcal{E}_{ T}$ implies $R_1 \LOSEconv \mathcal{E}_{ T}$, and hence $R_1 \LOSEconv R_2$ by transitivity.
Hence we have proven that the two orderings are the same; that is, $R_1 \succeq_{\rm LOSE}R_2$ if and only if $R_1 \succeq_{\mathcal{G}_{\! T}}R_2$. 
\end{proof}

\section{Methods for determining if a given channel is or is not LOSE-free} \label{neccond}

Ref.~\cite{causallocaliz} gives (in Theorem~5 therein) a necessary condition for a channel to be LOSE-free:
\begin{thm} \label{neccond1}
If $\mathcal{E}$ is an LOSE-free superoperator on $\mathcal{H}_A \otimes \mathcal{H}_B$, and $\ket{\psi}$, $A \otimes I \ket{\psi}$, and $I \otimes B \ket{\psi}$ are all eigenstates of $\mathcal{E}$ (where $A$ and $B$ are invertible operators), then $A\otimes B \ket{\psi}$ is also an eigenstate of $\mathcal{E}$.
\end{thm}

This theorem was leveraged in Ref.~\cite{causallocaliz} (as well as above) to prove that the BGNP channel is postquantum.

%
%

Ref.~\cite{Beckmanthesis} gives (in Theorem~3 therein) a sufficient condition for a channel to be LOSE-free:
\begin{thm}
Any bipartite channel $\mathcal{E}$ that can be implemented using local operations and one-way classical communication from Alice to Bob and that has an eigenstate of Schmidt rank equal to the dimension of Alice's input system is LOSE-free.
\end{thm}

Finally, note that if one can establish that a given distributed channel outperforms all LOSE strategies at some operational task, then it must be nonfree. This is true even for operational tasks that are not subsumed within our framework of distributed games, as long as one allows the players of the game to optimize over all LOSE strategies, since the score in such tasks constitutes an LOSE-monotone.
For example, in Ref.~\cite{causallocaliz}, the BGNP channel is shown to be a valuable postquantum resource by virtue of reducing the communication complexity of the inner product function below what could be achieved by LOSE operations.
In resource theoretic language: the amount of quantum communication required to simulate a postquantum resource (optimized over all LOSE operations) is a monotone~\cite{Gonda2019}, and the value assigned to the BGNP channel by this monotone is larger than can be achieved by any free resource. This is another proof that it is not LOSE-free.

\end{document}